\documentclass[aps,pra,twocolumn,notitlepage,longbibliography,nofootinbib]{revtex4-1}

\usepackage{graphicx}
\usepackage{amsmath,amssymb,amsthm,amsfonts}

\usepackage{mathrsfs}

\usepackage{xcolor}

\usepackage{comment}

\usepackage{thmtools,thm-restate}

\theoremstyle{plain}

\newtheorem{lemma}{Lemma}
\theoremstyle{definition}

\newtheorem{protocol}{Protocol}

\def\bF{\mathbb{F}}
\def\Si{\mathop{\rm Si}}

\newcommand{\ket}[1]{|#1\rangle}
\newcommand{\Tr}{{\rm Tr}\,}
\newcommand{\red}[1]{{\color{black}#1}}

\newcommand{\sld}{ {\mathrm{SLD}}}
\newcommand{\rld}{ {\mathrm{RLD}}}
\newcommand{\tr}{ {\mathrm{Tr}}}

\usepackage{hyperref}
\hypersetup{
colorlinks=true,
linkcolor=blue,
filecolor=blue,
citecolor=blue,  
urlcolor=blue,
}

\def\Label#1{\label{#1}\ [\ \text{#1}\ ]\ }
\def\Label{\label}

\begin{document}
\title{Global Heisenberg scaling in noisy and practical phase estimation}
\author{Masahito Hayashi}
\email{hayashi@sustech.edu.cn}
\affiliation{Shenzhen Institute for Quantum Science and Engineering, Guangdong Provincial Key Laboratory of Quantum Science and Engineering, Southern University of Science and Technology, Shenzhen, 518055, China}
\affiliation{Graduate School of Mathematics, Nagoya University, Furocho, Chikusa-ku, Nagoya 464-860, Japan}
\author{Zi-Wen Liu}
\email{zliu1@perimeterinstitute.ca}
\affiliation{Perimeter Institute for Theoretical Physics, Waterloo, Ontario  N2L 2Y5, Canada}
\author{Haidong Yuan}
\email{hdyuan@mae.cuhk.edu.hk}
\affiliation{Department of Mechanical and Automation Engineering, The Chinese University of Hong Kong, Shatin, Hong Kong SAR, China}

\begin{abstract}
Heisenberg scaling characterizes the ultimate precision of parameter estimation enabled by quantum mechanics, which represents an important quantum advantage of both theoretical and technological interest. 
Here, we study the attainability of strong, global notions of Heisenberg scaling in the fundamental problem of phase estimation, from a practical standpoint.  A main message of this work is an asymptotic noise ``threshold'' for global Heisenberg scaling.   We first demonstrate that Heisenberg scaling is fragile to noises in the sense that it cannot be achieved in the presence of phase damping noise with strength above a stringent scaling in the system size.  Nevertheless, we show that when the noise does not exceed this threshold, the global Heisenberg scaling in terms of limiting distribution (which we highlight as a practically important figure of merit) as well as average error can indeed be achieved.   Furthermore, we provide a practical adaptive protocol using one qubit only, which achieves global Heisenberg scaling in terms of  limiting distribution under such noise. 
\end{abstract}

\maketitle
\section{Introduction}

The estimation of unknown parameters such as phases in quantum systems, which is also widely studied under the names of quantum metrology, sensing, interferometry etc.\ in recent years \cite{PhysRevLett.96.010401,Giovannetti2011,RevModPhys.89.035002,PhysRevD.23.1693}, is a problem of fundamental importance in quantum information science \cite{Kitaev,NC}, as well as an exciting technological frontier with promising potential for practical applications in wide-ranging scenarios involving high-precision measurements such as spectroscopy, gravitational wave detection, and atomic clocks \cite{PhysRevA.54.R4649,PhysRevLett.88.231102,RevModPhys.87.637}. 
A central observation of this area is that by utilizing quantum mechanical effects such as superposition, entanglement and squeezing, quantum estimation can potentially attain precision which scales as $n^{-1}$ where $n$ is the resource count (e.g.~the number of channel uses or the probing time), namely the Heisenberg scaling \cite{Giovannetti1330,PhysRevLett.96.010401}.  In contrast, one can only attain the scaling of $n^{-1/2}$ (also known as the shot-noise or standard quantum limit) with classical resources.  This indicates a significant quantum enhancement in metrology and estimation tasks, which is a representative type of practical advantages of quantum information technologies.

However, quantum systems are very susceptible to the realistically ubiquitous noise effects, which stand as a fundamental obstacle towards practical quantum applications \cite{NC,Preskill2018quantumcomputingin}. Therefore, a research direction of central importance is to understand the limitations of quantum information processing, especially to what extent the theoretically blueprinted quantum advantages can be maintained, when noises are taken into account.  Ideally, for the standard phase estimation problem, where we aim to estimate the phase $\theta$ in the signal unitary $U_\theta:= e^{i\theta \sigma_z}$, it is well known that the Heisenberg scaling can be achieved in various settings \cite{Giovannetti1330,Higgins2007}.
Nevertheless, the estimation precision is naturally expected to deteriorate under noise effects, leading us to the following important and highly nontrivial question: When can Heisenberg scaling still be achieved in the presence of noises?


\begin{figure}
        \centering
\includegraphics[width=0.9\linewidth]{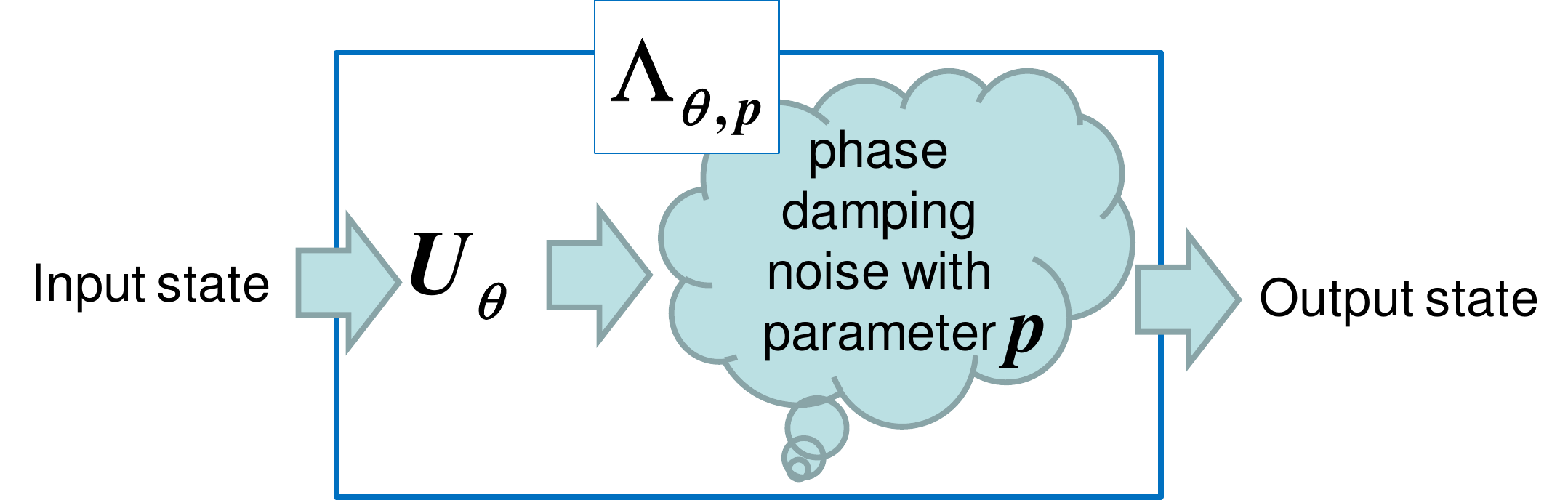}
\caption{Our model: The parameterized channel to be estimated $\Lambda_{\theta,p}$ is  the signal unitary $U_\theta$ affected by a phase damping noise
with parameter $p$.}
\Label{FF0}
\end{figure}

In this work, we address this general question by studying the necessary and sufficient conditions for achieving Heisenberg scaling in phase estimation, in the presence of the fundamental phase damping noise as illustrated in Fig. \ref{FF0}.
More specifically, we derive a strong upper bound on the noise strength, and further address the achievability when the bounds are satisfied by constructing explicit protocols.  Here, in particular, we consider a strong notion of estimation in terms of global precision over all possible values of the phase that is broadly important in practical applications, while most previous work only consider the local notion.  
Notably, the most widely studied lens for quantum metrology is the quantum Fisher information (QFI), 
which nevertheless only characterizes the local estimation precision at a given point and is generally insufficient for scenarios in which global estimation is of interest (see more detailed discussions later).  

The key contributions of this work are more specifically summarized as follows.  We first formally lay down two sets of natural criteria for global Heisenberg scaling, respectively based on the average error and the notion of limiting distribution \cite{IH09}.  In particular, the limiting distribution is a powerful notion that provides more information than the error measures commonly considered in metrology, allowing us to directly analyze confidence intervals and success probabilities. However, the study of it in quantum metrology is very limited (see also \cite{YCH}).
By explicitly analyzing the behavior of QFI under phase damping, we derive a $O(n^{-1})$ upper bound on the noise strength, which is necessary for Heisenberg scaling.  
On the other hand, when this bound is satisfied, we show that both notions of global Heisenberg scaling can indeed be achieved (a key tool being Fourier analysis), indicating that the $O(n^{-1})$ bound is optimal in a strong sense.   We also construct a practically friendly protocol that resorts to only single-qubit memories by modifying the well known phase estimation algorithm in \cite{Cleve} and show that it achieves global Heisenberg scaling in terms of limiting distribution. 
Note that previous work \cite{ZhouJiang21} implies that Heisenberg scaling cannot be achieved under any fixed strength of phase damping.  Here we extend the consideration to $n$-dependent noise to sharpen this understanding, and also first present protocols that actually achieves Heisenberg scaling under phase damping. Also note that our protocols are not based on quantum error correction as is commonly considered (see e.g.~\cite{PhysRevLett.112.150801,PhysRevLett.112.150802,PhysRevLett.112.080801,2013arXiv1310.3432O,PhysRevX.7.041009,Zhou2018,PhysRevLett.122.040502,Gorecki2020optimalprobeserror}) and thus broadens the methodology for quantum metrology in noisy scenarios. 


\section{Criteria for global Heisenberg scaling}

Here we discuss our global notions of Heisenberg scaling in detail.

As mentioned, a commonly considered but limited figure of merit for quantum metrology is the quantum Fisher information (QFI).
More specifically, the symmetric logarithmic derivative (SLD) QFI is given by $J_{\theta}^{\sld}=\tr(\rho_{\theta}L_S^2)$, where  $\rho_{\theta}$ is the the state carrying the parameter $\theta$ and $L_S$ is the SLD operator which can be obtained from the equation 
   $ \frac{\partial \rho_{\theta}}{\partial \theta}=\frac{1}{2}(L_S\rho_{\theta}+\rho_{\theta}L_S)$.
Then the quantum Cram\'{e}r-Rao bound gives a lower bound on the estimation error as measured by the standard deviation in terms of QFI \cite{Helstrom,HolevoP}: $\delta \hat{\theta}\geq \frac{1}{\sqrt{m J_{\theta}^{\sld}}}$, 
where $\delta \hat{\theta}=\sqrt{E[( \hat{\theta}-\theta)^2]}$ is the standard deviation, and $m$ is the number of times that the measurement is repeated.  Here, importantly, $\hat{\theta}$ is assumed to be an unbiased estimator (whose expected value equals the true value).
In the literature, the Heisenberg scaling is often considered in terms of the QFI scaling as $n^2$ where $n$ is the number of channel uses, as this indicates that $\delta \hat{\theta}$ scales as $1/n$ due to the quantum Cram\'{e}r-Rao bound. 
However, the QFI only bounds the local precision at a single point, while global notions that consider all possible values of the parameter are often important and more meaningful {as the true value of the parameter is supposed to be unknown}. 
The optimal local estimator in general does not work globally, as previously pointed out in e.g.~\cite{H,H2}.  In fact, even in a neighborhood of $\theta_0$, it does not work with respect to the minimax criterion (where one considers the worst point in the neighborhood) when the radius of the neighborhood of $\theta_0$ is a constant \cite{H}.
When the minimum \red{mean square} error of local estimation scales as $O(n^{-1})$,
it can be attained globally by using various adaptive methods including two-step methods \cite{H}.
However, the proof of the reduction statement does not work when the scaling is $O(n^{-1-\delta})$
for any $\delta>0$. 
Furthermore, it is known that in the parallel scheme
the minimum error for global phase estimation can be strictly larger than the inverse of the maximum QFI \cite{HVK,H}. 
This shows the necessity of a new method for global estimation.
Therefore, the $n^2$ scaling of QFI does not mean that it is possible to construct an estimator that can achieve the Heisenberg scaling globally, even with adaptive estimation.  We refer interested readers to e.g.~\cite{HVK,H,H2} for more discussions on this issue. 

We would like to rigorously study the attainability of global notions of Heisenberg scaling, for which it is not sufficient to consider QFI (although it can lead to simple necessary conditions, as will be discussed later).
Here we consider two types of figure of merit. The first is the average error over all possible values of the parameter.  
For our phase estimation problem where $\theta\in(-\pi,\pi]$, 
considering periodicity, we focus on e.g.~the error function 
$\tilde{R}_{\theta}: = \mathbb{E}_{\hat\theta}[ \sin^2(\hat\theta-\theta)]$, where
$\mathbb{E}_{\hat\theta}$ denotes the expectation with respect to $\hat\theta$.
Then, we take its average with respect to the uniform prior distribution over the range of $\theta$:
\begin{align}
\tilde{R}:=\frac{1}{2\pi} \int_{-\pi}^{\pi}  R_{\theta} d\theta,
\end{align}
We say the Heisenberg scaling is achieved when $R$ scales as $1/n$. 
The second figure of merit, which is practically more important but little understood, is the 
probability that the error exceeds a certain threshold $c$, namely 
$P_{\theta} \{ |\hat{\theta}-\theta| > c\}$.
When the threshold $c$ is a constant, this is just the large deviation analysis \cite{LD}.
Here we are interested in the case where the limiting probability is constant. This is in general only possible when the threshold $c$ changes with $n$ and 
the Heisenberg scaling means
the threshold $c$ has scaling $O(n^{-1})$.
To be more precise, we say that the Heisenberg scaling  in terms of limiting distribution is achieved if ${P}_{\theta} \Big\{
\frac{a}{n} \le \hat\theta -\theta \le \frac{b}{n} \Big\}$ converges to 
an non-trivial value (neither $0$ nor $1$) 
for any two real numbers $a<b$, in this case we can define the limiting distribution $\bar{P}_\theta$ as
$\bar{P}_\theta(a,b):=\lim_{n \to \infty}P_{\theta} \{
\frac{a}{n} \le \hat\theta -\theta \le \frac{b}{n} \}$.
The limiting distribution is more informative about the estimation
as it can be used to calculate the error probability exceeds a certain threshold $c$, $P_{\theta} \{ |\hat{\theta}-\theta| > c\}$,
which is widely used in practice.
Note that the global Heisenberg scalings under these two figures of merit are slightly different, as will be seen later.

\section{Global phase estimation under noise}
We now present our results on the attainability of global Heisenberg scaling in the presence of noise.
We consider a model where the signal unitary is given by
$U_\theta:= e^{i \theta/2}|0\rangle \langle 0|
+e^{-i \theta/2}|1\rangle \langle 1|$  (where $\theta \in (-\pi,\pi]$)
on the system ${\cal H}$ spanned by $\{|0\rangle,|1\rangle\}$, and there is a phase damping noise $\mathcal{N}_p(\rho) = (1-p)\rho + p |0\rangle \langle 0| \rho |0\rangle \langle 0| 
+p |1\rangle \langle 1| \rho |1\rangle \langle 1|$ with dephasing probability or strength $p \in [0,1]$ which describes the natural decoherence effect, acting before or after the application of
$U_\theta$.  Noting that the signal unitary $U_\theta$ acts trivially upon dephasing, our model is overall given by the channel
\begin{align}
\Lambda_{\theta,p}(\rho):= (1-p)U_\theta\rho U_\theta^\dagger
+p |0\rangle \langle 0| \rho |0\rangle \langle 0| 
+p |1\rangle \langle 1| \rho |1\rangle \langle 1|. 
\end{align}

\begin{figure}
        \centering
\includegraphics[width=0.9\linewidth]{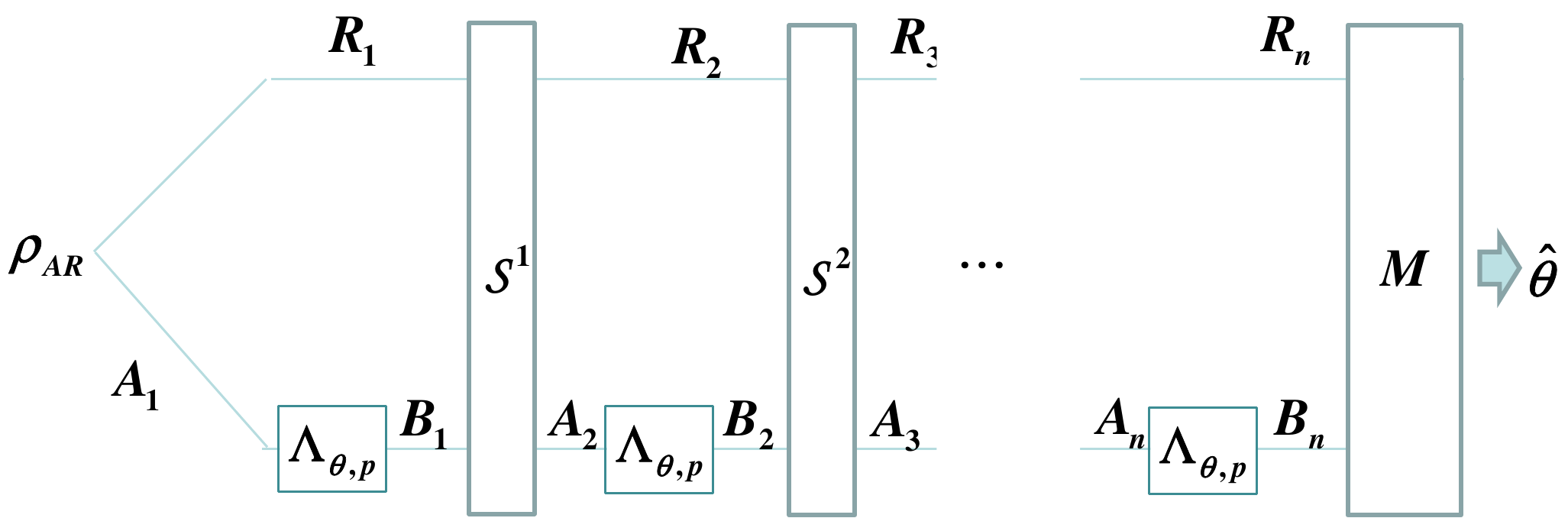}
\caption{Adaptive scheme:
$A_i$ is the input system of the $i$-th channel application.
$B_i$ is the output system of the $i$-th channel application.
$R_i$ is the memory system during the $i$-th channel application.
$\mathcal{S}^i$ is the channel that connects the $i$-th and $(i+1)$-th channel applications.
We apply $n$ uses of channel $\Lambda_{\theta,p}$ in an adaptive way, which represents the most general approach to channel parameter estimation. 
The $n$ uses of the channel $\Lambda_{\theta,p}$ 
are interleaved with $n-1$ quantum
channels $\mathcal{S}^1, \cdots, \mathcal{S}^{n-1}$, which can also share memory systems with each other. 
The final measurement $M$ outputs 
the outcome $\hat{\theta}$ 
as our estimate of the unknown parameter $\theta$. }
\Label{FF1}
\end{figure}

Here, we study the Heisenberg scaling of the channel estimation of
$\Lambda_{\theta,p}$ under the adaptive scheme, as illustrated in Fig.~\ref{FF1}, which represents the most general approach to channel estimation.
Since the $n^2$ scaling of the QFI is a necessary condition for global Heisenberg scaling 
in terms of both figures of merit, we can obtain a simple upper bound on the noise strength $p$ through analyzing QFI as follows. 
Note that the SLD QFI is upper bounded by the right logarithmic derivative (RLD) QFI, namely $J_{\theta}^{\sld} \leq J_{\theta}^{\rld}$ where 
$J_{\theta}^{\rld}=\tr(L_R^\dagger \rho_{\theta} L_R)$ is the RLD QFI and $L_R$ is the RLD operator satisfying
   $ \frac{\partial \rho_{\theta}}{\partial \theta}=\rho_{\theta}L_R$.   
   We denote the SLD (RLD) QFI of the output state of $\Lambda_{\theta,p}$ acting on input state $\rho$  as ${J}_{\theta,p, \rho}^{\sld(\rld)}$, and then the channel SLD (RLD) QFI of  $\Lambda_{\theta,p}$  given by maximizing over all input states as $\mathcal{J}_{\theta,p}$, namely $\mathcal{J}_{\theta,p}^{\sld(\rld)} := \max_\rho {J}_{\theta,p, \rho}^{\sld(\rld)}$.  
Although the maximum SLD QFI $\mathcal{J}_{\theta,p}^{\sld}$ is not additive which makes the analysis of it difficult in general,
the maximum RLD QFI $\mathcal{J}_{\theta,p}^{\rld}$ is additive even in the adaptive scheme \cite[Theorem 18]{Katariya}. 
So we need only address the RLD QFI to derive a necessary condition for Heisenberg scaling.
It can be verified that $\mathcal{J}_{\theta,p}^{\rld}=\frac{2(1-p)^2}{p(2-p)}$. Due to the additivity of RLD QFI,  
the maximum RLD QFI 
for $n$ uses
then equals $n\frac{2(1-p)^2}{p(2-p)}$. 
Therefore, the maximum RLD QFI scales as $n^2$ if and only if
$p\leq O(n^{-1})$.  See Appendix A for more detailed calculations and discussions.

Note that this $p\leq O(n^{-1})$ bound on the phase damping strength is quite strong, comparable to e.g.\ an erasure noise model in which only a constant number of qubits are erased in a scalable system of $n$ qubits.  
In fact, it is easy to check that $p$ would need to be sub-constant to achieve any scaling advantage over the shot-noise limit. 
This is consistent with (and improves) the previous knowledge \cite{ZhouJiang21} that Heisenberg scaling cannot be achieved for any constant $p>0$ in our noise model due to the ``Hamiltonian-not-in-Kraus-span'' condition.  An overall message is that the metrological advantage of quantum systems is highly fragile in noisy environments.


\smallskip

{Now we consider whether the conditions for global Heisenberg scaling can actually be attained when $p\leq O(n^{-1})$ (see Appendix B for a detailed exposition).  
To set the stage , we first discuss the noiseless model where $n$ unitary channels act in parallel on a $n$-qubit input state, which is assumed to be a pure state
$|\eta\rangle := \sum_{m=0}^{n} a_m\ket{m}$, where 
$\ket{m}$ is a normalized vector in the eigenspace of 
$\frac{d}{d\theta}U_\theta^{\otimes n}|_{\theta=0}$ with eigenvalue 
$\frac{n}{2}-m$. 
We choose the coefficients $a_m$ as $a_m:= \frac{1}{\sqrt{n+1}}f(\frac{m}{n})$, here $f$ is some square-integrable $C^1$-differentiable function on $[0,1]$ with $l^2$ norm $1$, which is the key object in our analysis.
The distribution of the outcome of the phase covariant measurement
is then given by
\begin{align*}
&P_{\theta} \Big\{
\frac{a}{n} \le \hat\theta -\theta \le \frac{b}{n} \Big\}
=
P_{0} \Big\{
\frac{a}{n} \le \hat\theta  \le \frac{b}{n} \Big\}\\
=&
\int_{\frac{a}{n} }^{\frac{b}{n} }
\frac{1}{n+1}
\Big|\sum_{m=0}^{n} 
e^{i \hat\theta m}f(\frac{m}{n}) \Big| ^2 \frac{d \hat\theta}{2\pi}
\cong  \int_{a}^{b}
|{\cal F}f(t)|^2 dt ,
\end{align*}
where ${\cal F}f$ {denotes} the Fourier transform of $f$.
That is, the limiting distribution of the estimate is determined by 
the Fourier transform ${\cal F}f$ \cite{IH09}, and 
global Heisenberg scaling in terms of limiting distribution can be achieved 
when the input state $|\eta\rangle$ is given by any square-integrable $C^1$-differentiable function $f$ on $[0,1]$ with $l^2$ norm equals to $1$.
As for the average error,
consider $R[|\eta\rangle]:= \mathbb{E}_{\hat\theta,\theta}[\sin^2 (\hat\theta -\theta)]$ for input state $|\eta\rangle$,
where the error function is taken to be $\sin^2 (\hat\theta -\theta)$.}
Suppose the Dirichlet boundary condition, i.e.~$f(0)=f(1)=0$, holds, e.g., 
$f(x)$ is given by $\sqrt{2} \sin (\pi x)$.
Then we have
\begin{align}
\tilde{R}[|\eta\rangle]=\frac{1}{4 n^2}\langle f |P^2|f \rangle 
+o\left(\frac{1}{n^2}\right),\label{NKA6}
\end{align}
where $P=-i \frac{d}{dx}$.
When the Dirichlet boundary condition does not hold, 
$\tilde{R}[|\eta\rangle]= O(n^{-1})$;
More specifically,
\begin{align}
\tilde{R}[|\eta\rangle] \cong 
\frac{1}{n}(A_+(f)+A_-(f))\Si(2 \pi),
\end{align}
where 
\begin{align}
A_+(f)&=\lim_{R_1 \to \infty}\lim_{R_2 \to \infty}\frac{1}{R_2}\int_{R_1}^{R_1+R_2}
t^2|{\cal F}f(t)|^2 dt \\
A_-(f)&=\lim_{R_1 \to -\infty}\lim_{R_2 \to -\infty}\frac{1}{R_2}\int_{R_1}^{R_1+R_2}
t^2|{\cal F}f(t)|^2 dt,
\end{align}
$\Si(x):= \int_0^x \frac{\sin t}{t}dt$, and 
$\Si(2 \pi) \cong 1.41815$.
Therefore, we conclude that the average error condition for global Heisenberg scaling is achieved if and only if 
the Dirichlet boundary condition $f(0)=f(1)=0$ holds.

For a concrete case, consider the input state with the form $|\eta_{\rm uni}\rangle:= \sum_{m=0}^n
\frac{1}{\sqrt{n+1}} |m\rangle$, where
$f$ takes constant value $1$ on $[0,1]$ and 
we have
$\tilde{R}[|\eta_{\rm uni}\rangle]
=  \frac{\Si(2\pi)}{2 \pi n}+O(\frac{1}{n^2})$. 
This state achieves the global Heisenberg scaling in terms of limiting distribution, but not the average error, demonstrating that these two conditions are not equivalent.

In the presence of $p=\frac{\epsilon}{n}$ noise, 
the above analyses for the limit distribution and the average error $R[|\eta\rangle]$
are changed as follows.
For given integers $k,\ell$,
we define the operator $T_{t,k,\ell}$ as
\begin{align}
&T_{t,k,\ell}\nonumber \\
:=& \sum_{u=\max(0,t-k+\ell)}^{\min(t,l)} 
{k-\ell \choose t-u} {\ell \choose u} Q^{2(t-u)+\ell}(I-Q)^{2u+k-\ell} ,
\end{align}
where $Q$ is the multiplication operator.
Then, the average error is calculated under the the Dirichlet boundary condition $f(0)=f(1)=0$ as
\begin{align}
\tilde{R}[|\eta \rangle]
 \cong 
\sum_{k=0}^\infty
e^{-\epsilon}\frac{\epsilon^k}{k !}
\sum_{t=0}^{k}
\sum_{\ell=0}^k {k \choose \ell} 
\frac{1}{4n^2}\langle f| \sqrt{T_{t,k,\ell}}P^2\sqrt{T_{t,k,\ell}}|f\rangle.
\end{align}
Since the Dirichlet boundary condition for $f$ implies 
the Dirichlet boundary condition for $\sqrt{T_{t,k,\ell}} f$,
the average error $\tilde{R}[|\eta \rangle]$ achieves the Heisenberg scaling even in the case with noise 
$p=\frac{\epsilon}{n}$.
As for the limiting distribution condition, we have  
\begin{align}
& 
P_{\theta}
\Big\{
\frac{a}{n} \le \hat\theta -\theta \le \frac{b}{n} \Big\}\nonumber \\
\cong &
\sum_{k=0}^\infty
e^{-\epsilon}\frac{\epsilon^k}{k !}
\sum_{t=0}^{k}
\sum_{\ell=0}^k {k \choose \ell} 
\int_{a}^{b}
|{\cal F}(\sqrt{T_{t,k,\ell}}f)(t)|^2 dt .\Label{ZXO2}
\end{align}
That is, we find that 
the Heisenberg scaling in terms of limiting distribution can be achieved
even when $f$ does not satisfy the Dirichlet boundary condition. The overall message is summarized as follows.
\begin{restatable}{theorem}{thmglobal}\Label{TH1}
\red{The strength of phase damping noise $p\in O(1/n)$ is a necessary and sufficient condition for the existence of an estimator to achieve global Heisenberg scaling in terms of both average error and limiting distribution.}
\end{restatable}

\section{A practical method using single-qubit memory}

{In the above, we demonstrated the attainability of the global Heisenberg scaling with $n$ channels acting in parallel on a $n$-qubit state. However, the protocol is practically demanding  since the state is in general highly entangled and the measurement typically needs to be collective. In the following we propose and analyze a simple adaptive one-qubit protocol that builds on the phase estimation algorithm in \cite{Cleve} (see Appendix C for details).}

{\begin{protocol}\Label{P1}
In the first step, 
we prepare the input state $|+\rangle:=
\frac{1}{\sqrt{2}}(|0\rangle+|1\rangle)$,
and apply the unknown channel $\Lambda_{\theta,p}$ for $2^N$ times.
Then, we measure the final state in the basis
$\{|+\rangle, |-\rangle \}$ and set
$A_1=0, 1$ upon getting $|+\rangle, |-\rangle$ respectively.


Inductively, in the $k$-th step, 
we prepare the input state $|+\rangle:=
\frac{1}{\sqrt{2}}(|0\rangle+|1\rangle)$,
and apply  $\Lambda_{\theta,p}$ for $2^{N-k+1}$ times.
Then, we apply $U_{-A_1  2^{-k+1} \pi-A_2 2^{-k+2}\pi - \cdots -A_{k-1} 2^{-1}\pi  }$ depending on $A_1, \cdots,A_{k-1}$.
Then, we measure the final state in the basis
$\{|+\rangle, |-\rangle \}$ and set
$A_k=0, 1$ upon getting $|+\rangle, |-\rangle$ respectively.

We repeat the above up to the $(N+1)$-th step.
After the final step, depending on $A^{N+1}:=(A_1, \cdots, A_{N+1})$, we obtain the final estimate
$\hat{\theta}(A^{N+1}):= A_1  2^{-N} \pi-A_2 2^{-N+1}\pi + \cdots + A_k 2^{k-(N+1)}\pi
+ \cdots +A_{N}2^{-1} \pi +A_{N+1} \pi$. 

This protocol uses $n:=2^{N+1}-1 $ applications of the unknown channel $\Lambda_{\theta,p}$ in total.
\end{protocol}}
   

For the noiseless case, the stochastic behavior of the estimate $\hat\theta$
turns out to be the same as the $\eta=\eta_{\rm uni}$ case above .
The noisy case requires a different analysis.
Again, consider $p=\frac{\epsilon}{n}$ noise. 
Then, the stochastic behavior of the error $\hat{\theta}-\theta$
is asymptotically characterized as
\begin{align}
& \lim_{N \to \infty}\mathbb{E}_\theta [P_{\theta,\frac{\epsilon}{n},n} 
\Big\{ \frac{a}{n}\le \hat{\theta}-\theta \le \frac{b}{n}\Big\}] \nonumber \\
=&\lim_{N \to \infty}
\mathbb{E}_{\hat{A}^{N+1},X^{N+1}}
\bigg[\int_{b}^{a}
 \frac{\sin^2 y }{(y+\zeta(\hat{A}^{N+1},X^{N+1}))^2}\frac{dy}{2\pi}\bigg].
\label{EY1}
\end{align}
Here, 
the term $\zeta(\hat{A}^{N+1},X^{N+1}):= 
((-1)^{\hat{A}_1} 2 X_1   + \cdots +(-1)^{\hat{A}_{N+1}}2^{N+1} X_{N+1})\pi $ 
represents the difference from the noiseless case.
In addition,
the binary variables 
$\hat{A}^{N+1}:=(\hat{A}_1,  \cdots, \hat{A}_{N+1})$ are independent binary variables subject to the uniform distribution
and 
the binary variables
$X^{N+1}:=(X_1,  \cdots, X_{N+1})$ 
are independently subject to the following distribution:
\begin{align}
P_{X_k}(1)=\frac{1}{2}\big(1-e^{-\epsilon 2^{-k}}\big),\quad
P_{X_k}(0)=\frac{1}{2}\big(1+e^{-\epsilon 2^{-k}}\big).
\Label{AJ1}
\end{align}
Therefore, 
the proposed estimator
achieves Heisenberg scaling in terms of limiting distribution.
For intuitions, we show in Fig.~\ref{FF2} a numerical comparison between the PDFs for the limiting distributions of the noiseless and noisy cases.
Also, the asymptotics of the average error is explicitly calculated to be
\begin{align}
\mathbb{E}[(\hat{\theta}-\theta)^2] \cong \frac{1}{n} \big( \epsilon \pi^2
+\frac{\Si(2\pi)}{ \pi }\big).
\end{align}


\begin{figure}
        \centering
 \includegraphics[width=0.7\linewidth]{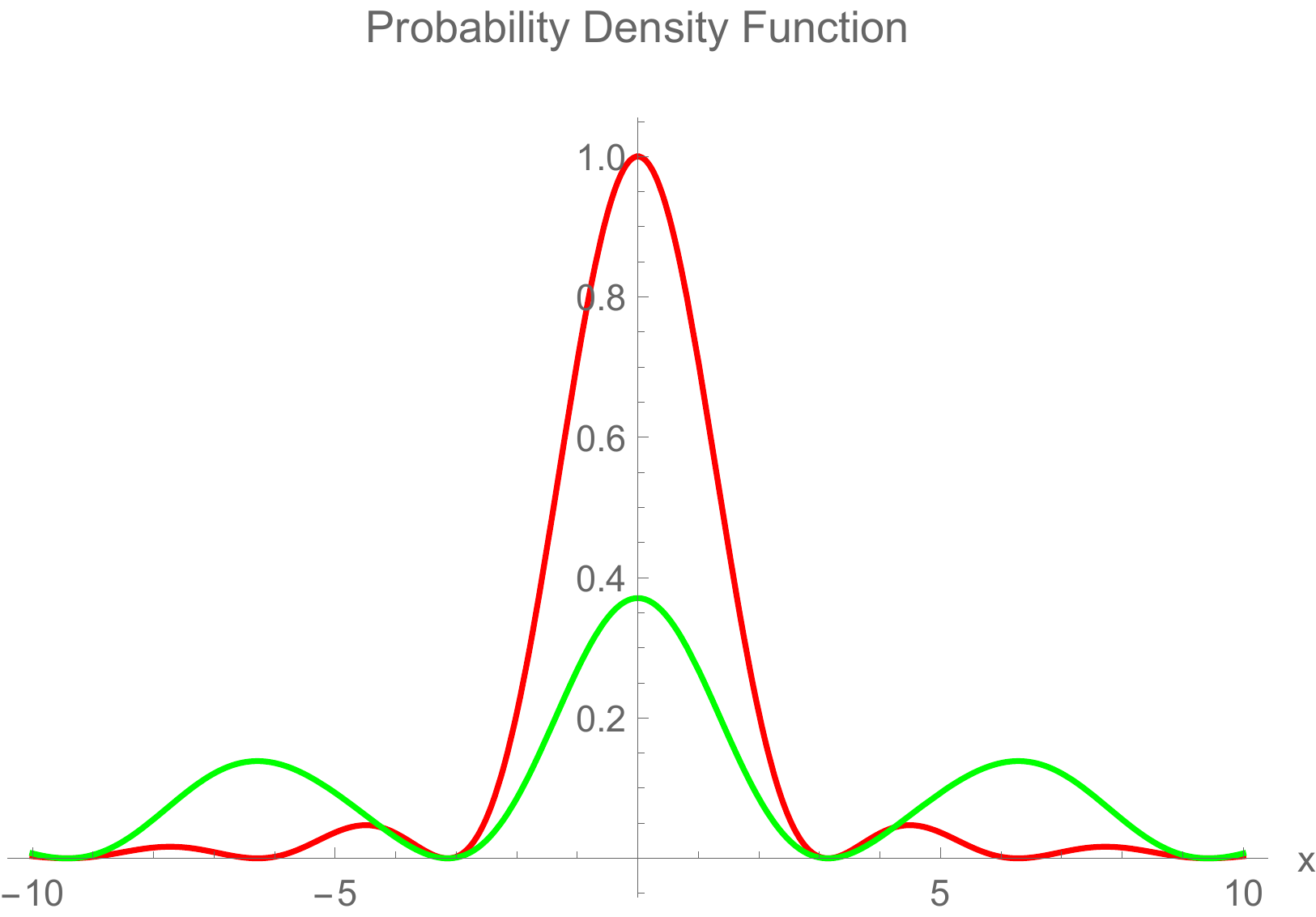}
 \caption{A numerical comparison between the limiting distributions of the noiseless and noisy cases.
The red curve is the probability density function $\frac{\sin^2 x}{x^2}$
 of the noiseless case. 
The green curve is the probability density function
 of the noisy case with $\epsilon=1$. 
}
\Label{FF2}
\end{figure}

\section{Summary and outlook}
In this work, we considered the problem of achieving Heisenberg scaling in phase estimation in a global sense, for which the previously widely studied Cram\'{e}r-Rao approach and quantum Fisher information is not sufficient.  We introduced two types of meaningful conditions for such global Heisenberg scaling, respectively based on the average error and the limiting distribution.  
In particular, we consider the limiting distribution to be a practically important and powerful tool, and we hope our work stimulates more interest in this perspective.  Here we are particularly interested in the attainability of global Heisenberg scaling in practical scenarios where noise effects are present.  We considered phase damping noise and proved a necessary and sufficient condition on the noise strength which can be regarded a strong ``threshold theorem'' for global Heisenberg scaling -- when the noise strength is above the $O(n^{-1})$ scaling we showed by analyzing QFI that the Heisenberg scaling cannot be achieved;  while otherwise, we gave a protocol that achieves both notions of global Heisenberg scaling.  Furthermore, we generalized the well known phase estimation algorithm in \cite{Cleve} to construct a practically implementable protocol that uses only a small memory, which achieves global Heisenberg scaling in terms of limiting distribution.   

For future work, it would be interesting to extend the analysis to more general noise models such as depolarizing and erasure noises, and consider the estimation of more general actions like SU($d$).  Furthermore, note that phase estimation is an essential subroutine in a wide range of important quantum algorithms (e.g.~factoring \cite{Shor}, linear system solving \cite{HHL}).  
In praticular, the limiting distribution method enables us to consider success probability, which is important in practice.   An important future work is to explore the connections and applications of our results to the practical implementation of such algorithms.

\section*{Acknowledgements}
MH is supported in part by Guangdong Provincial Key
Laboratory (grant no. 2019B121203002).
ZWL is supported by Perimeter Institute for Theoretical Physics.
Research at Perimeter Institute is supported in part by the Government of Canada through the Department of Innovation, Science and Economic Development Canada and by the Province of Ontario through the Ministry of Colleges and Universities.  Part of this work was done during ZWL's visit to SUSTech quantum institute and ZWL would like to thank the institute for hospitality.

\bibliography{Noisy-estimation}

\appendix

\setcounter{theorem}{0}
\setcounter{lemma}{0}
\setcounter{figure}{0}
\renewcommand{\thefigure}{S\arabic{figure}}
\renewcommand{\thelemma}{S\arabic{lemma}}
\renewcommand{\thetheorem}{S\arabic{theorem}}
\renewcommand{\thecorollary}{S\arabic{corollary}}
\renewcommand{\theprotocol}{S\arabic{protocol}}

\newtheorem{innercustomthm}{Protocol}
\newenvironment{protocolc}[1]
  {\renewcommand\theinnercustomthm{#1}\innercustomthm}
  {\endinnercustomthm}

\section{Quantum Fisher information under noise}\label{Sec:necc}
Here we give details and extended discussions on the analysis of QFI.  The main result is that Heisenberg scaling can only be achieved when noise strength $p\leq O(n^{-1})$.  As a further note, it is easy to verify that $p$ would need to be sub-constant to achieve any scaling advantage over the shot-noise limit. 

The following variant called the right logarithmic derivative (RLD) QFI will be useful: 
$J_{\theta}^{\rld}=\tr(L_R^\dagger \rho_{\theta} L_R)$ where  $L_R$ is the RLD operator which can be obtained from
\begin{align}
    \frac{\partial \rho_{\theta}}{\partial \theta}=\rho_{\theta}L_R.
\end{align}
  Note that the RLD QFI is an upper bound on the SLD QFI, namely $J_{\theta}^{\sld} \leq J_{\theta}^{\rld}$.  We shall be interested in various QFIs associated with our model channel $\Lambda_{\theta,p}$.  We denote the SLD (RLD) QFI of the output state of $\Lambda_{\theta,p}$ acting on input state $\rho$  as ${J}_{\theta,p, \rho}^{\sld(\rld)}$, and then the channel SLD (RLD) QFI of  $\Lambda_{\theta,p}$  given by maximizing over all input states as $\mathcal{J}_{\theta,p}$, namely $\mathcal{J}_{\theta,p}^{\sld(\rld)} := \max_\rho {J}_{\theta,p, \rho}^{\sld(\rld)}$.  
  {We shall also consider the maximum SLD (RLD) QFI 
  under $n$ uses of  $\Lambda_{\theta,p}$ in the parallel or adaptive schemes, respectively denoted by 
  $\mathcal{J}_{\theta,p}^{\sld(\rld),(n)}$ or $\mathcal{J}_{\theta,p}^{\sld(\rld),[n]}$.}

\subsection{RLD QFI}\label{app:rld}

The channel RLD QFI $\mathcal{J}_{\theta,p}^{\rld}$ can be computed from the Choi matrix.
Here the Choi matrix of $\Lambda_{\theta,p}$ is given by
   \begin{align}
C_{\Lambda_{\theta,p}}
=& 2(1-p)|\Phi_\theta\rangle \langle \Phi_\theta|+p (|0,0\rangle \langle 0,0|+|1,1\rangle \langle 1,1|) \nonumber\\
=& (2-p)|\Phi_\theta\rangle \langle \Phi_\theta|+p |\Phi_\theta^\perp\rangle \langle \Phi_\theta^\perp|
\end{align}
where $\Phi_\theta:= \frac{1}{\sqrt{2}}( e^{i \theta/2}|0,0\rangle+ e^{-i \theta/2}|1,1\rangle)$
and
$\Phi_\theta^\perp:= \frac{i}{\sqrt{2}}( e^{i \theta/2}|0,0\rangle- e^{-i \theta/2}|1,1\rangle)$.
The derivative $D_{\Lambda_{\theta,p}}:=\frac{d}{d\theta}C_{\Lambda_{\theta,p}}$ is 
\begin{align}
D_{\Lambda_{\theta,p}}= (1-p) (|\Phi_\theta^{\perp}\rangle \langle \Phi_\theta|+|\Phi_\theta\rangle \langle \Phi_\theta^{\perp}|).
\end{align}
Then  by using the formula \cite[Theorem 1]{H},  we obtain that
   \begin{align}
&
\mathcal{J}_{\theta,p}^{\rld} 
=
\| \Tr_{\mathrm{Out}} D_{\Lambda_{\theta,p}} C_{\Lambda_{\theta,p}}^{-1}D_{\Lambda_{\theta,p}}\|\nonumber \\
=&
\Big\| \Tr_{\mathrm{Out}} 
\Big(
\frac{(1-p)^2}{p} |\Phi_\theta\rangle \langle \Phi_\theta|+
\frac{(1-p)^2}{2-p} |\Phi_\theta^\perp\rangle \langle \Phi_\theta^\perp|
\Big)
\Big\|\nonumber \\
=&
\Big\| 
\frac{2(1-p)^2}{p(2-p)} I 
\Big\|
= \frac{2(1-p)^2}{p(2-p)} ,
\end{align}
where ``$\mathrm{Out}$'' denotes the output system of the channel, and $\|X\|$ denotes the matrix norm of $X$.
   
For the parallel scheme with $n$ uses of $\Lambda_{\theta,p}$, i.e.~$\Lambda_{\theta,p}^{\otimes n}$,
we simply have \cite[Corollary 1]{H} 
\begin{align}
  \mathcal{J}_{\theta,p}^{\rld,(n)} = n   \mathcal{J}_{\theta,p}^{\rld}=n\frac{2(1-p)^2}{p(2-p)}.   \Label{eq:rld}
\end{align}  
Recently, it has been shown that  for the general adaptive scheme (Fig. \ref{FF1})
the RLD QFI is additive \cite[Theorem 18]{Katariya}, so we again have 
\begin{align}
    \mathcal{J}_{\theta,p}^{\rld,[n]} = n  \mathcal{J}_{\theta,p}^{\rld} = n\frac{2(1-p)^2}{p(2-p)}.
\label{MNA}
\end{align}
Therefore, to achieve Heisenberg scaling, it is necessary that $p\leq O(n^{-1})$.   In particular, when $p=\frac{\epsilon}{n}$, we have
\begin{align}
    \mathcal{J}_{\theta,p}^{\rld,(n)} = \mathcal{J}_{\theta,p}^{\rld,[n]} = \frac{1}{\epsilon}n^2 + O(n).
\end{align}

\subsection{Achievable SLD QFI}
Now since RLD QFI upper bounds SLD QFI, we conclude that the standard SLD QFI cannot achieve the Heisenberg scaling if $p$ tends to zero more slowly than the order $O(n^{-1})$.   

We are now going to show that when $p\leq O(n^{-1})$ is satisfied, the SLD QFI can indeed achieve the Heisenberg scaling as well as the RLD QFI.  This can be seen by considering the GHZ state $|\Psi^{(n)}\rangle:=
\frac{1}{\sqrt{2}}(|0,\cdots,0\rangle+|1,\cdots,1\rangle)$ as the input state.
First consider the parallel scheme.
Even in the presence of the phase damping noise, the state belongs to the subspace ${\cal H}_o$ spanned by 
$|0,\cdots,0\rangle$ and $|1,\cdots,1\rangle$.
When  dephasing  acts at least on one qubit,
the state becomes the completely mixed state on ${\cal H}_o$.
That is,
\begin{align}
& \Lambda_{\theta,p}^{\otimes n}
(|\Psi^{(n)}\rangle \langle\Psi^{(n)}|)
\nonumber\\
=&
(1-p)^n|\Psi_\theta^{(n)}\rangle\langle \Psi_\theta^{(n)}|
+
\frac{(1-(1-p)^n)}{2}I_{{\cal H}_o} \nonumber \\
=&
\frac{(1+(1-p)^n)}{2} |\Psi_\theta^{(n)}\rangle\langle \Psi_\theta^{(n)}|
\nonumber\\&+
\frac{(1-(1-p)^n)}{2}|\Psi_\theta^{\perp,(n)}\rangle\langle \Psi_\theta^{\perp,(n)}|,
\end{align}
where 
$|\Psi_\theta^{(n)}\rangle:= \frac{1}{\sqrt{2}}(e^{i n\theta/2}|0,\cdots,0\rangle+e^{-in \theta/2}|1,\cdots,1\rangle)$
and
$|\Psi_\theta^{\perp,(n)}\rangle:= \frac{1}{\sqrt{2}}(e^{i n\theta/2}|0,\cdots,0\rangle-e^{-in \theta/2}|1,\cdots,1\rangle)$.
The SLD is given by
\begin{align}
n(1-p)^n
(|\Psi^{(\perp,n)}_\theta\rangle \langle\Psi^{(n)}_\theta|
+|\Psi^{(n)}_\theta\rangle \langle\Psi^{\perp,(n)}_\theta|).
\end{align}
{Hence, the SLD QFI of the state family $\Lambda_{\theta,p}^{\otimes n}
(|\Psi^{(n)}\rangle \langle\Psi^{(n)}|)$ is $n^2(1-p)^{2n}$, which implies that 
 \begin{align}
  \mathcal{J}_{\theta,p}^{\sld,(n)}\ge  
n^2(1-p)^{2n}.
\end{align}
Since $  \mathcal{J}_{\theta,p}^{\rld,(n)}\ge \mathcal{J}_{\theta,p}^{\sld,(n)}$,
combining with Eq.~\eqref{MNA}, for $p = \frac{\epsilon}{n}$
we have  
 \begin{align}
\mathcal{J}_{\theta,p}^{\sld,(n)}= n^2 [e^{-2\epsilon}+o(1)].
\label{LGU}
\end{align}

However, we stress again that the analysis of QFI only {guarantees an understanding of local precision}.  
When we employ the optimal local estimator,  
the mean square error behaves as $1/[n^2 (e^{-2\epsilon}+o(1))]= e^{2\epsilon}/n^2+o(1/n^2)$,
but this estimator in general does not work globally or even in certain neighborhood of the point, as pointed out in \cite{H}.

\subsection{SLD QFI for a practical adaptive strategy}\label{NNP}
Since $\mathcal{J}_{\theta,p}^{\rld,[n]} \ge\mathcal{J}_{\theta,p}^{\sld,[n]} \ge \mathcal{J}_{\theta,p}^{\sld,(n)}$,
the maximum SLD QFI $\mathcal{J}_{\theta,p}^{\sld,[n]}$ in the adaptive scheme has 
the same asymptotic behavior as Eq.~\eqref{LGU}.
Here we show that this asymptotic behavior can be achieved by a simple adaptive strategy on one qubit.
We consider $n$ repetitive applications of the channel $\Lambda_{\theta,p}$, i.e.,
 $\Lambda_{\theta,p}^{\circ n}:=
\underbrace{\Lambda_{\theta,p}\circ \cdots \circ\Lambda_{\theta,p} }_{n}$, acting on the input state $|+\rangle:=
\frac{1}{\sqrt{2}}(|0\rangle+|1\rangle)$. 
The output state is given by
\begin{align}
\Lambda_{\theta,p}^{\circ n}(|+\rangle \langle+|)
=
(1-p)^n |\Psi_{\theta,n}\rangle \langle\Psi_{\theta,n}|
+
(1-(1-p)^n )\frac{I}{2},
\end{align}
where 
$|\Psi_{\theta,n}\rangle:= \frac{1}{\sqrt{2}}(e^{i n\theta/2}|0\rangle+e^{-in \theta/2}|1\rangle)$
and
$|\Psi_\theta^{\perp}\rangle:= \frac{1}{\sqrt{2}}(e^{i n\theta/2}|0\rangle-e^{-in \theta/2}|1\rangle)$.
Then, the SLD QFI of the family $\Lambda_{\theta,p}^{\circ n}(|+\rangle \langle+|)$
is again calculated to be $n^2(1-p)^{2n}$.
When $p=\frac{\epsilon}{n}$, it is 
$n^2 [e^{-2\epsilon}+o(1)]$, namely the Heisenberg scaling of SLD QFI is achieved in the same way as Eq.~\eqref{LGU}.}
Again the optimal estimator also only works locally. 
A key finding of our work is that this estimator can be modified to achieve the Heisenberg scaling globally.

\section{Global phase estimation}\Label{Sec:parallel}
We have shown that 
a necessary condition for Heisenberg scaling is $p=O(n^{-1})$, and the main goal here is to prove that 
the Heisenberg scaling can be  achieved globally when the noise parameter $p$ behaves as $\epsilon/n$.  Along the way,  a comprehensive analysis of global phase estimation in both noiseless and noisy cases is given.

Before diving into the derivations, we overview the state of knowledge.
Since the Cram\'{e}r-Rao approach only addresses the precision of local estimation,
we need new methods to study the achievability of Heisenberg scaling for global estimation.
It is known in the noiseless case that global phase estimation can be done using the notion of group covariant estimators \cite{Holevo,HolevoP}. 
Recall that we are interested in two formulations for Heisenberg scaling in global estimation, respectively based on the
 asymptotics of the average error and the limiting distribution.
There exists a type of estimators that achieve the Heisenberg scaling in terms of limiting distribution but not  average error in the noiseless case (note that the asymptotic behavior of the average error is not known previously).
The previous study \cite{IH09} discusses how Heisenberg scaling can be achieved in both senses in the noiseless case.  The noisy case has not been studied before.

This section aims to provide a detailed, self-contained discussion of global estimation.
As a preparation, we first discuss the noiseless case in Appendix \ref{S7-A}. Appendix \ref{S7-A-1} introduces the group covariant formulation for
global phase estimation.
 Appendix \ref{S7-A-2} reviews the existing results for 
global Heisenberg scaling in the noiseless case.
In Appendix \ref{S7-A-3}, as a new result,
we explicitly analyze the asymptotics of the average error of the estimators
that achieve the Heisenberg scaling in terms of limiting distribution but not average error.
Finally, in Appendix \ref{SEMA}, we consider the noisy case and show how Heisenberg scaling is achieved in both senses
when the noise parameter $p= \epsilon/n$.

\subsection{Noiseless case}\Label{S7-A}
\subsubsection{Formulation with covariant measurements}\Label{S7-A-1}
For the parallel scheme with $n$ uses of the unknown unitary,  
it is known that the inverse of the maximum SLD Fisher information cannot be  attained globally in general \cite{H,H2}. 
In our case with $p=0$, the minimum average error is strictly larger than the inverse of the maximum SLD Fisher information \cite{LP,BDM,H2,H2}. 
In these previous studies \cite{LP,BDM,H2}, 
the global phase estimation is achieved by converting 
the parallel operation $U_\theta^{\otimes n}$ on $n$-qubits to a phase operation 
$U_{\theta}^{(n)}:=
\sum_{m=0}^{n} e^{i (2m -n)\theta/2}|m\rangle \langle m|$ 
on an $(n+1)$-dimensional system ${\cal H}_n$
spanned by $\{|0\rangle, |1\rangle,\cdots, |n\rangle\}$.

To see this conversion, we define 
the subspace ${\cal H}_{n,m}$ of ${\cal H}^{\otimes n}$ 
spanned by the vector $|0\rangle^{\otimes n-m}\otimes |1\rangle^{\otimes m}$ 
and its permutations with respect to the order of the tensor product. 
The initial state $|\Xi\rangle$ on ${\cal H}^{\otimes n}$ can be decomposed 
as $|\Xi\rangle = \sum_{m=0}^n a_m |\Xi_{n,n}\rangle$, where $|\Xi_{n,m}\rangle \in {\cal H}_{n,m}$
is a normalized vector and $a_m$ is a non-negative real number.
Depending on the states $\vec{\Xi}_n:= \{ |\Xi_{n,m}\rangle\}_{m=0}^n$,
we define the isometry $V_{\vec{\Xi}_n}$ from ${\cal H}_n$ to ${\cal H}^{\otimes n}$ as 
\begin{align}
V_{\vec{\Xi}_n}( |m\rangle):=|\Xi_{n,m}\rangle.
\end{align}
Then, $U_{\theta}^{(n)}$ is given as $
V_{\vec{\Xi}_n}^\dagger U_\theta^{\otimes n} 
V_{\vec{\Xi}_n}$.

So for the noiseless case,
the problem of estimating the unknown unitary $U_\theta^{\otimes n} $ on ${\cal H}^{\otimes n}$
with the initial state $|\Xi\rangle$ 
is converted to estimating the unknown unitary $U_\theta^{(n)} $ on ${\cal H}_n$
with the initial state $|\eta\rangle:=
\sum_{m=0}^n a_m |m\rangle$ \cite{Holevo,HolevoP}.
Now suppose the error function $R(\theta,\hat{\theta}) $
depending on the true parameter $\theta$ and the estimate $\hat{\theta}$ has the property
\begin{align}
R ( \theta+\theta',\hat{\theta}+\theta')=R(\theta,\hat{\theta}) \label{Cov}
\end{align}
for any $\theta' $.
Then when the estimator is given by a POVM  $M$ on ${\cal H}_n$, 
the error is given by 
\begin{align}
R_{M,\theta}:=\int_0^{2 \pi} R(\theta,\hat{\theta})  
\langle \tilde{\Xi}|M(d\hat{\theta}) |\tilde{\Xi}\rangle .
\end{align}
Then the  average error is naturally defined by taking the average with respect to the uniform prior distribution of the true parameter $\theta$:
\begin{align}
R_{M}:=\frac{1}{2\pi} \int_{-\pi}^{\pi}  R_{M,\theta} d\theta.\label{prior}
\end{align}
If the estimator $M$ satisfies the covariance condition,
\begin{align}
U_c^{(n)} M((a,b)) (U_c^{(n)})^{\dagger}= M((a+c,b+c)),
\end{align}
it is called a covariant estimator.  Then
the error $R_{M,\theta}$ does not depend on the true parameter $\theta$, namely for any $\theta$,
\begin{align}
R_{M,\theta}= R_{M}.
\end{align}
Given an estimator $M$, we can define an associated covariant estimator $\bar{M}$ as
\begin{align}
\bar{M}( (a,b)) := \frac{1}{2\pi}
\int_{-\pi}^{\pi} U_c^{(n)} M((a-c,b-c)) (U_c^{(n)})^{\dagger} dc,
\end{align}
which satisfies the condition
\begin{align}
R_{M}= R_{\bar{M},\theta}= R_{\bar{M}}\label{C1}
\end{align}
for any $\theta$.
For example, consider the discrete estimator $M_{{\rm Dis}}[\{e^{i \theta_j}\}_j]$ 
\begin{align}
&M_{{\rm Dis}}[\{e^{i \theta_j}\}_j]
(\hat{\theta}=\frac{2 j \pi}{n})\nonumber \\
=&
U_{\frac{2 j \pi}{n}}^{(n)} |I,\{e^{i \theta_j}\}_j\rangle 
\langle I,\{e^{i \theta_j}\}_j| (U_{\frac{2 j \pi}{n}}^{(n)} )^\dagger
\end{align}
where $|I,\{e^{i \theta_j}\}_j\rangle:= \sum_{j=0}^n e^{i \theta_j}|j\rangle$.
The covariant estimator $\bar{M}_{{\rm Dis}}[\{e^{i \theta_j}\}_j]$ 
is equivalent to the continuous estimator $M_{\rm Con}[\{e^{i \theta_j}\}_j]$ defined as
\begin{align}
&M_{\rm Con}[\{e^{i \theta_j}\}_j](d\hat{\theta})\nonumber \\
:=& 
U_{\hat{\theta}}^{(n)} |I,\{e^{i \theta_j}\}_j\rangle 
\langle I,\{e^{i \theta_j}\}_j| (U_{\hat{\theta}}^{(n)} )^\dagger
\frac{d \hat{\theta}}{2\pi}.
\end{align}
Therefore, when we evaluate the average error $R_{M}$ of a given estimator $M$,
we can consider the error $R_{\bar{M},\theta}$ of the associated covariant estimator $\bar{M}$.
In particular, the minimization problem can be simplified as 
\begin{align}
\min_{M} R_{M}=\min_{M: {\rm covariant}} R_{M,\theta}.
=\min_{M: {\rm covariant}} R_{M}.
\end{align}
That is, it is sufficient to minimize over covariant estimators \cite{Holevo,HolevoP}.

\subsubsection{Heisenberg scaling of average error and limiting distribution}\Label{S7-A-2}
For channel estimation, it is known that
the local minimax error  can be  asymptotically achieved globally \cite{H}.
This fact shows that the asymptotic performance does not depend on the choice of the prior distribution
on the parameter space.
Therefore, 
without loss of the generality, we may assume the uniform prior distribution (Eq.~\eqref{prior})
in later asymptotic discussions. 

Now, we adopt the common error function 
$\tilde{R}(\theta,\hat{\theta}) := \sin^2 (\theta-\hat{\theta})$.
It is known that the minimum $\tilde{R}_{M}$ is given by \cite{Holevo,HolevoP,H2}
\begin{align}
&\min_{M} \tilde{R}_{M}\nonumber \\
= &\tilde{R}[|\eta\rangle]
:=
\int_{-\pi}^{\pi }
|\langle F_{\hat{\theta}-\theta}^\dagger I |  \eta\rangle |^2 
\sin^2 (\hat{\theta}-\theta)
\frac{d\hat{\theta}}{2\pi}\nonumber  \\
=&\frac{1}{2} \Big(\sum_{m=0}^{n}|a_m|^2
-\frac{1}{2}\sum_{m=0}^{n-1} \bar{a}_{m+1}a_m 
-\frac{1}{2}\sum_{m=1}^{n} \bar{a}_{m-1}a_m 
\Big)
\nonumber \\
=&\langle \vec{a}| T |\vec{a}\rangle,
\label{M1}
\end{align}
where
the matrix $T$ is defined as
$T_{k,m}=\frac{1}{2}\delta_{k,m}-\frac{1}{4}\delta_{k,m+1}-\frac{1}{4}\delta_{k,m-1}$,
which can be attained when the estimator $M$ is taken to be
the continuous estimator $M_{\rm Con}$.
Here,
the covariant measurement is given by
\begin{align}
M(d\hat{\theta})= \frac{d\hat{\theta}}{2\pi} 
U_{\hat{\theta}}^{(n)}\Big(\sum_{m=0}^n  e^{i \theta_m} |m\rangle\Big)
\Big(\sum_{m=0}^n \langle m| e^{-i \theta_m}\Big)
U_{-\hat{\theta}}^{(n)}.
\end{align}
 When $e^{i \theta_j}=\frac{a_j}{|a_j|}$, the above measurement is
 an optimal one that achieves Eq.~\eqref{M1}.
Due to Eq.~\eqref{C1}, the minimum is achieved by an estimator $M$
when the associated covariant estimator $\bar{M}$ is $M_{\rm Con}
[\{e^{i \theta_j}\}_j]$.
Hence, the discrete estimator $M_{{\rm Dis}}
[\{e^{i \theta_j}\}_j]$ achieves the minimum
because $\bar{M}_{{\rm Dis}}[\{e^{i \theta_j}\}_j]$ is
$M_{\rm Con} [\{e^{i \theta_j}\}_j]$.
 
For the estimation of the unknown unitary $U_{\theta}^{(n)}$, we have
\begin{align}
    \min_{|\eta\rangle}\tilde{R}[|\eta\rangle] = \frac{1}{2}  (1-\cos \frac{\pi}{n+1} ) = \frac{\pi^2}{4 (n+1)^2}+o(\frac{1}{n^2}),
\end{align}
which is achieved when 
$|\eta\rangle$ is chosen to be
\begin{align}
    |\eta_{\rm opt}\rangle:= \sum_{m=0}^n
C \sin \frac{\pi m}{n+1} |m\rangle,
\end{align}
 where $C$ is the normalization constant \cite{LP}\cite[Eq. (10)]{BDM}\cite[Theorem 7]{H3}. 
That is, $|\eta_{\rm opt}\rangle$ achieves the Heisenberg scaling in terms of average error. As a contrast, if the initial state is taken to be
$|\eta_{\rm uni}\rangle:= \sum_{m=0}^n
\frac{1}{\sqrt{n+1}} |m\rangle$,
the error is given by $R[|\eta_{\rm uni}\rangle] = \frac{1}{n+1}$, namely the Heisenberg scaling is not achieved.
 
In fact, the minimum coefficient $\frac{\pi^2}{4}$ in the Heisenberg scaling can be derived in another way as follows.
Consider a square-integrable $C^2$-differentiable function $f$ on $[0,1]$ with the $l^2$ norm $1$. 
We choose the coefficients $a_m:= \frac{1}{\sqrt{n+1}}f(\frac{m}{n})$ for the input state $|\eta\rangle$.
When the Dirichlet boundary condition $f(0)=f(1)=0$ is satisfied,
depending on $|\eta\rangle$,
the average $\tilde{R}[|\eta\rangle]$ is calculated as follows:
\begin{align}
\tilde{R}[|\eta\rangle]=& \frac{1}{4(n+1)}
\bigg(\sum_{m=1}^{n-1}\bar{f}(\frac{m}{n})
\Big(2{f}(\frac{m}{n})-{f}(\frac{m+1}{n})\nonumber \\
&-{f}(\frac{m-1}{n})\Big)+|{f}(0)|^2+|{f}(1)|^2
\bigg) \nonumber \\
=& \frac{-1}{4(n+1)n^2}\sum_{m=1}^{n-1}
\bar{f}(\frac{m}{n})
\frac{1}{n}\Big(\frac{f(\frac{m+1}{n})-f(\frac{m}{n})}{\frac{1}{n}} \nonumber \\
&-\frac{f(\frac{m}{n})-f(\frac{m-1}{n})}{\frac{1}{n}}
\Big) \nonumber \\
= & 
\frac{-1}{4(n+1)n^2}\sum_{m=1}^{n-1}
\bar{f}(\frac{m}{n})
\frac{d^2 f }{dx^2}(\frac{m}{n})+o(\frac{1}{n^2})
\nonumber  \\
= & 
\frac{-1}{4 n^2}
\int_{0}^{1}
\bar{f}(x)
\frac{d^2 f }{dx^2}(x) dx +o(\frac{1}{n^2})
\nonumber  \\
=&
\frac{1}{4 n^2}\langle f |P^2|f \rangle +o(\frac{1}{n^2}),\label{NKA}
\end{align}
where $P=-i \frac{d}{dx}$.
Under the Dirichlet boundary condition $f(0)=f(1)=0$,
the minimum eigenvalue of $P^2$ is $\pi^2$ and
the corresponding eigenfunction is $\sqrt{2} \sin (\pi x)$.
On the other hand,
when the Dirichlet boundary condition is not satisfied,
the first equation does not hold and  
$\langle f |P^2|f \rangle$ does not take a finite value. 
Therefore, we need a different analysis for this case (see Appendix~\ref{S7-A-3}).

We now go on to consider the practically more important figure of merit, the probability $P_{\theta,0,n} \{ |\hat{\theta}-\theta| > c\}$ where $c$ is a certain error threshold (we use $P_{\theta,p,n}$ to denote the distribution when
the true parameter, 
the dephasing probability,
and the number of applications,
are $\theta$, $p$, and $n$, respectively).
As mentioned in the main text, the case of constant $c$ corresponds to the large deviation analysis \cite{LD}, and here
we consider the case where
the limiting probability is a constant.
When Heisenberg scaling is achieved, 
the threshold $c$ has scaling $O(n^{-1})$.
Hence, when 
$P_{\theta,0,n} \{
\frac{a}{n} \le \hat\theta -\theta \le \frac{b}{n}\}$ converges to 
an non-trivial value, i.e., neither $0$ nor $1$,
we say that the limiting distribution
$ \lim_{n \to \infty}P_{\theta,0,n} \{
\frac{a}{n} \le \hat\theta -\theta \le \frac{b}{n} \}$ achieves the Heisenberg scaling.
In the following, we show that 
the Heisenberg scaling in terms of limiting distribution can be achieved  without
the Dirichlet boundary condition. 
For this discussion, we denote the Fourier transform of $f$ by ${\cal F}f$, which is defined as
\begin{align}
{\cal F}f(t):= \frac{1}{\sqrt{2\pi}}\int_{-\infty}^{\infty}e^{i x t}f(x) dx.
\end{align}
Then, using $t=n(\hat{\theta}-\theta)$, we have
\begin{align}
& 
P_{\theta,0,n} \Big\{
\frac{a}{n} \le \hat\theta -\theta \le \frac{b}{n} \Big\}\nonumber \\
=&
\int_{\frac{a}{n}}^{\frac{b}{n} }
\Big|\sum_{m=0}^n e^{i m (\hat{\theta}-\theta) -i n (\hat{\theta}-\theta)/2} a_m\Big|^2
\frac{d\hat{\theta}}{2\pi} \nonumber \\
=&\frac{1}{n+1}\int_{\frac{a}{n}}^{\frac{b}{n} }
\Big|\sum_{m=0}^n e^{i m (\hat{\theta}-\theta)} f(\frac{m}{n})\Big|^2
\frac{d\hat{\theta}}{2\pi} \nonumber \\
=&(n+1)\int_{\frac{a}{n}}^{\frac{b}{n} }
\Big|\frac{1}{n+1}\sum_{m=0}^n e^{i \frac{m}{n} n(\hat{\theta}-\theta)} f(\frac{m}{n})\Big|^2
\frac{d\hat{\theta}}{2\pi}\nonumber \\
\cong &\frac{(n+1)}{n}\int_{a}^{b}
|{\cal F}f(t)|^2 dt 
\cong \int_{a}^{b}
|{\cal F}f(t)|^2 dt .\Label{ZXO}
\end{align}
That is, we can say that 
the Heisenberg scaling  in terms of limiting distribution can be achieved \cite{IH09}.

Now return to the original problem of estimating the unknown unitary $U_{\theta}^{\otimes n}$
on the system ${\cal H}^{\otimes n}$.
The choice of $f$, or $\{a_i\}$, corresponds to the choice of the initial pure state.
The above analysis shows that the minimum average error as given by the error function
$\tilde{R}$ is given by $\frac{1}{2}  (1-\cos \frac{\pi}{n+1} )
= \frac{\pi^2}{4 (n+1)^2}+o(\frac{1}{n^2})$ and thus achieves Heisenberg scaling.
Furthermore, it can be attained by initial state 
$V_{\vec{\Xi}_n}^\dagger |\eta_{\rm opt}\rangle$
and the POVM
\begin{align}
M(d\hat{\theta})= 
V_{\vec{\Xi}_n}^\dagger
M_{\rm Con}(d\hat{\theta})V_{\vec{\Xi}_n}.
\end{align}

\subsubsection{Asymptotic analysis of average error without the Dirichlet boundary condition}\Label{S7-A-3}
The above discussion shows that the Heisenberg scaling  in terms of 
average error  is not achieved
when the square-integrable $C^1$-differentiable function $f$ on $[0,1]$ does not satisfy
the Dirichlet boundary condition $f(0)=f(1)=0$.
However,  the asymptotic behavior of the average error $R[|\eta\rangle]$
in this case has not been explicitly analyzed.
We now do so.
Using $t=n(\hat{\theta}-\theta)$, we have
\begin{align}
&\tilde{R}[|\eta\rangle]\nonumber \\=&
\int_{-\pi}^{\pi }
\left|\sum_{m=0}^n e^{i m (\hat{\theta}-\theta) -i n (\hat{\theta}-\theta)/2} a_m\right|^2
\sin^2(\hat{\theta}-\theta)
\frac{d\hat{\theta}}{2\pi} \nonumber \\
=&(n+1)\int_{-\pi}^{\pi }
\left|\frac{1}{n+1}\sum_{m=0}^n e^{i \frac{m}{n} n(\hat{\theta}-\theta)} f(\frac{m}{n})\right|^2
\nonumber \\
&\quad\cdot\sin^2( \frac{n(\hat{\theta}-\theta)}{n})
\frac{d\hat{\theta}}{2\pi}\nonumber \\
\cong &\frac{n+1}{n}\int_{-\pi n}^{\pi n }
\sin^2( \frac{t}{n})
|{\cal F}f(t)|^2 dt .
\end{align}
When $f$ satisfies the Dirichlet boundary condition $f(0)=f(1)=0$,
the integral $\int_{-\infty}^{\infty }
t^2|{\cal F}f(t)|^2 dt$ converges.
Hence, we have
\begin{align}
\lim_{n \to \infty} n^2 \tilde{R}[|\eta\rangle]
=&
\lim_{n \to \infty}
\int_{-\pi n}^{\pi n }
n^2 \sin^2( \frac{t}{n})
|{\cal F}f(t)|^2 dt \nonumber \\
=&\int_{-\infty}^{\infty }
t^2|{\cal F}f(t)|^2 dt.
\end{align}
To consider a function $f$ that does not satisfy the Dirichlet boundary condition $f(0)=f(1)=0$,
we define 
\begin{align}
A_+(f)&:=\lim_{R_1 \to \infty}\lim_{R_2 \to \infty}\frac{1}{R_2}\int_{R_1}^{R_1+R_2}
t^2|{\cal F}f(t)|^2 dt \\
A_-(f)&:=\lim_{R_1 \to -\infty}\lim_{R_2 \to -\infty}\frac{1}{R_2}\int_{R_1}^{R_1+R_2}
t^2|{\cal F}f(t)|^2 dt.
\end{align}
We then have
\begin{align}
\tilde{R}[|\eta\rangle]=&\frac{n+1}{n}\int_{-\pi n}^{\pi n }
\sin^2( \frac{t}{n})
|{\cal F}f(t)|^2 dt \nonumber \\
=&\frac{n+1}{n^2}\int_{-\pi}^{\pi}
\frac{\sin^2 s}{s^2}
|{\cal F}f(\frac{s}{n})|^2(\frac{s}{n})^2 ds \nonumber \\
\cong & \frac{1}{n}
\Big(A_+(f) \int_{0}^{\pi}\frac{\sin^2 s}{s^2}  ds 
+A_-(f)\int_{-\pi}^{0}\frac{\sin^2 s}{s^2} ds \Big)
\nonumber \\=& 
\frac{1}{n}(A_+(f)+A_-(f))\Si(2 \pi),
\end{align}
where $\Si(x):= \int_0^x \frac{\sin t}{t}dt$, and 
$\Si(2 \pi) \cong 1.41815$.
That is, this type of input states cannot achieve Heisenberg scaling in terms of average error.
However, we shall see that it achieves Heisenberg scaling in terms of limiting distribution.

A representative example of an input state that does not satisfy the Dirichlet boundary condition is the state $|\eta_{\rm uni}\rangle$, for which
$f$ is the constant function on $[0,1]$.
Now since ${\cal F}f(t)= \frac{e^{it/2} \sin \frac{t}{2} }{\sqrt{2\pi} t} $, we have
$t^2 |{\cal F}f(t)|^2= \frac{\sin^2 \frac{t}{2}}{2\pi}$, which implies $A_+(f)=A_-(f)=\frac{1}{4\pi}$.
Hence, we have
\begin{align}
\tilde{R}[|\eta_{\rm uni}\rangle]
\cong  \frac{\Si(2\pi)}{2 \pi n},\Label{XO1}
\end{align}
which has a scaling  different from the Heisenberg scaling.

\subsection{Noisy case}\Label{SEMA}
We now extend the analysis to the noisy scenario. 
In the following parts, we employ the standard notation for probability theory, in which, 
upper case letters denote random variables
and the corresponding lower case letters denote their realizations.
Consider the tensor-product vector space $(\mathbb{C}^2)^{\otimes n}$, where $\mathbb{C}^2$ is spanned by 
the normalized orthogonal basis $\{|0\rangle,|1\rangle\}$.
The tensor-product vector space $(\mathbb{C}^2)^{\otimes n}$ can be decomposed as
\begin{align*}
(\mathbb{C}^2)^{\otimes n}=\bigoplus_{j=0 \hbox{ or } 1/2 }^{n/2} {\cal U}_{j} \otimes {\cal V}_j,
\end{align*}
where ${\cal U}_{j}$ denotes the spin $j$ representation of SU(2), and
${\cal V}_{j}$ denotes the irreducible representation of $n$-th permutation group with respect to the order of tensor product.
${\cal V}_{j}$ is spanned by $\{|\ell, j\rangle\}_{\ell=-j,-j+1, \cdots, j-1,j}$ and we denote the projection to ${\cal U}_{j} \otimes {\cal V}_j$ as $P_{n,j}$.

Now, we consider the phase estimation problem under our noisy model $\Lambda_{\theta, p}$.
In order to make the noise effect symmetric with respect to permutation,
we let $|\Xi_{n,m}\rangle={n \choose m}^{-1/2} (|1\rangle^{\otimes m}\otimes |0\rangle^{\otimes n-m}+ PT)$,
where $PT$ represents the permuted terms of $|1\rangle^{\otimes m}\otimes |0\rangle^{\otimes n-m}$.
An initial state which is permutation invariant can then be written as $|\Phi\rangle:=\sum_{m=0}^n a_m |\Xi_{n,m}\rangle$.
For simplicity, we let the unitary $U_{\theta}^{\otimes n}$ act on the $n$ qubits  after the phase damping channel. 
Let $X^n=(X_1, \cdots, X_n) \in \bF_2^n$ be the variables that describe the effects of the noise: When the dephasing, i.e., 
the two-valued measurement $\{|0\rangle \langle 0|,  |1\rangle \langle 1|\}$ is applied on the $i$-th qubit,  $X_i=1$, otherwise $X_i=0$.
Let $|\vec{x}|$ be the number of 
components in the vector $\vec{x}$ that are $1$.
For example, when $\vec{x}=(\underbrace{1,\cdots,1}_{k},\underbrace{0,\cdots,0}_{n-k})$, we have $|\vec{x}|=k$.

In the following, we consider the above type of $\vec{x}$.
Since the dephasing acts on the first $k$ qubits,
the PVM $\{P_{\vec{z}|\vec{x}}\}_{\vec{z}\in \bF_2^k}$ is applied, where
the projection $P_{\vec{z}|\vec{x}}$ is defined as 
$|\vec{z}\rangle \langle \vec{z}| \otimes I^{\otimes n-k}$.
For a general $\vec{x}\in \bF_2^n$, 
the projection $P_{\vec{z}|\vec{x}}$ is defined 
by applying the permutation to 
$|\vec{z}\rangle \langle \vec{z}| \otimes I^{\otimes n-k}$.
Therefore, 
when $\vec{X}=\vec{x}$, the resultant state is
\begin{align}
\sum_{\vec{z} \in \bF_2^{|\vec{x}|}} 
P_{\vec{z}|\vec{x}}
|\Phi\rangle \langle \Phi|P_{\vec{z}|\vec{x}}.
\end{align}
Noting that the probability that $\vec{X}=\vec{x}$ is $p^k (1-p)^{n-k}$,
the averaged state is
\begin{align}
\sum_{\vec{x} \in \bF_2^{n}} 
p^{|\vec{x}|}(1-p)^{n-|\vec{x}|}
\sum_{\vec{z} \in \bF_2^{|\vec{x}|}} 
P_{\vec{z}|\vec{x}}|\Phi\rangle \langle \Phi|P_{\vec{z}|\vec{x}}.\label{DD1}
\end{align}
Since this state is invariant with respect to permutation, it has the form
\begin{align}
\bigoplus_{j=0 \hbox{ or } 1/2 }^{n/2} p_j \rho_{j} \otimes 
I_{{\cal V}_j},
\label{DD2}
\end{align}
where $\rho_j$ is some state on ${\cal U}_j$
and $I_{{\cal V}_j}$ is the completely mixed state on ${\cal V}_j$.
Since the unitary $U_\theta^{\otimes n}$ acts only on ${\cal U}_j$,
the optimization of measurement is reduced to the phase estimation on each system 
${\cal U}_j$.
Since the basis $\{|\ell, j\rangle\}_{\ell=-j,-j+1, \cdots, j-1,j}$ of ${\cal V}_{j}$ are eigenvectors of 
the unitary $U_\theta^{\otimes n}$,
we apply the following measurement on ${\cal U}_{j}\otimes {\cal V}_{j}$,
\begin{align}
&M(d\hat{\theta})\nonumber \\
=& \frac{d\hat{\theta}}{2\pi} 
U_{\hat{\theta}}^{(n)}\Big(\sum_{\ell=-j}^j e^{i \theta_\ell} |\ell,j\rangle\Big)
\Big(\sum_{\ell=-j}^j \langle \ell,j| e^{-i \theta_\ell}\Big)
U_{-\hat{\theta}}^{(n)} \otimes I_{{\cal V}_{j}}.
\label{MFI}
\end{align}
Since the state Eq.~\eqref{DD2} can be decomposed into Eq.~\eqref{DD1},
we consider the optimal coefficient $e^{i \theta_\ell}$ for each component of the decomposition.

Without loss of generality, we consider the case when 
$\vec{x}=(\underbrace{1,\cdots,1}_{k},\underbrace{0,\cdots,0}_{n-k})$
and $\vec{z}=(\underbrace{1,\cdots,1}_{\ell},\underbrace{0,\cdots,0}_{k-\ell})$.
Then,
\begin{align}
P_{\vec{z}|\vec{x}}|\Phi\rangle 
=&\sum_{m=0}^n a_m P_{\vec{z}|\vec{x}} |\Xi_{n,m}\rangle
\nonumber \\=&\sum_{m=0}^n a_m 
b_{n,m|k, \ell}  |1\rangle^{\otimes \ell} |0\rangle^{\otimes (k-\ell)} |\Xi_{n-k,m-\ell}\rangle ,
\end{align}
where
\begin{align}
&b_{n,m|k, \ell} \nonumber \\
:=&
\sqrt{
\frac{m (m-1)\cdots (m-\ell+1)  }{
n (n-1)\cdots (n-k+1)} } \nonumber \\
& \cdot \sqrt{(n-m) (n-m-1)\cdots (n-m-(k-\ell)+1) }.
\Label{XI1}
\end{align}
We choose a non-negative coefficient $c_{n,j,m,k,\ell}$ and $| \Upsilon(n,j,k,\ell) \rangle
\in {\cal V}_{j}$ as
\begin{align}
 &P_{n,j} |1\rangle^{\otimes \ell} |0\rangle^{\otimes (k-\ell)} |\Xi_{n-k,m-\ell}\rangle\nonumber  \\
=&c_{n,j,m,k,\ell} |m,j\rangle\otimes | \Upsilon(n,j,k,\ell) \rangle
\end{align}
Here, the vector $| \Upsilon(n,j,k,\ell) \rangle\in {\cal V}_{j}$ is a normalized vector, and 
does not depend on $m$ because 
the operators $J_+ $ and $J_-$ commute with the projection $ P_{n,j}$.
The non-negativity of $c_{n,j,m,k,\ell}$ follows from that it is given as a summand of non-negative coefficients
based on a combinatorial discussion.
We have $c_{n,j,m,k,\ell}=0$ if and only if 
$\frac{n}{2}-j > \min(k,m,n-m) $.
Therefore,
\begin{align}
P_{\vec{z}|\vec{x}}|\Phi\rangle 
=
\bigoplus_{j=0 \hbox{ or } 1/2 }^{n/2}
\sum_{m=0}^n a_m d_{n,j,m,k,\ell} |m,j\rangle\otimes | \Upsilon(n,j,k,\ell) \rangle.
\end{align}
where
$d_{n,j,m,k,\ell}:=b_{n,m|k, \ell}
 c_{n,j,m,k,\ell}$.
Since $a_m d_{n,j,m,k,\ell} \ge 0$,
the optimal coefficient $e^{i \theta_\ell}$ is $1$.
This optimal choice does not depend on the component in the decomposition given in Eq.~\eqref{DD1}.

When the initial state before the application of the unknown phase is
$\frac{1}{\sqrt{\langle \Phi|P_{\vec{z}|\vec{x}}|\Phi\rangle }}
P_{\vec{z}|\vec{x}}|\Phi\rangle$
and the measurement given by Eq.~\eqref{MFI} is applied,
due to Eq.~\eqref{M1}, 
the estimation error of this estimator is
\begin{align}
&\frac{1}{{\langle \Phi|P_{\vec{z}|\vec{x}}|\Phi\rangle }}
\sum_{j=0 \hbox{ or } 1/2 }^{n/2}
\frac{1}{2} \Big(\sum_{m=0}^{n}a_m^2 d_{n,j,m,k,\ell} ^2 \nonumber \\
&\quad -\frac{1}{2}\sum_{m=0}^{n-1} a_{m+1} d_{n,j,m+1,k,\ell} a_m d_{n,j,m,k,\ell} \nonumber \\
&\quad
-\frac{1}{2}\sum_{m=1}^{n} a_{m-1} d_{n,j,m-1,k,\ell} a_m d_{n,j,m,k,\ell}
\Big).\label{MFU}
\end{align}
Hence, when $\vec{X}=\vec{x}$,  
the initial state before the application of the unknown phase is
$\sum_{\vec{z}\in \{0,1\}^k} P_{\vec{z}|\vec{x}}|\Phi\rangle\langle \Phi|P_{\vec{z}|\vec{x}}$.
In this case, when the measurement given by Eq.~\eqref{MFI} is applied,
by taking the average with respect to $ \vec{z}$ in Eq.~\eqref{MFU}, 
the estimation error of this estimator is
\begin{align}
&\sum_{j=0 \hbox{ or } 1/2 }^{n/2}
\sum_{\ell=0}^k {k \choose \ell}
\frac{1}{2} \Big(\sum_{m=0}^{n}a_m^2 d_{n,j,m,k,\ell} ^2\nonumber \\
&\quad
-\frac{1}{2}\sum_{m=0}^{n-1} a_{m+1} d_{n,j,m+1,k,\ell} a_m d_{n,j,m,k,\ell} 
\nonumber \\
&\quad-\frac{1}{2}\sum_{m=1}^{n} a_{m-1} d_{n,j,m-1,k,\ell} a_m d_{n,j,m,k,\ell}
\Big).
\end{align}
Finally, taking the average for $k$ under the distribution
${n \choose k}p^k(1-p)^{n-k}$,
the estimation error is 
\begin{align}
v_n:=&
\sum_{j=0 \hbox{ or } 1/2 }^{n/2}
\sum_{k=0}^n {n \choose k}p^k(1-p)^{n-k}\nonumber \\
&\cdot
\sum_{\ell=0}^k {k \choose \ell}
\frac{1}{2} \Big(\sum_{m=0}^{n}a_m^2 d_{n,j,m,k,\ell} ^2\nonumber \\
&\quad
-\frac{1}{2}\sum_{m=0}^{n-1} a_{m+1} d_{n,j,m+1,k,\ell} a_m d_{n,j,m,k,\ell} \nonumber \\
&\quad
-\frac{1}{2}\sum_{m=1}^{n} a_{m-1} d_{n,j,m-1,k,\ell} a_m d_{n,j,m,k,\ell}
\Big).
\end{align}
We need to minimize the above value by choosing $a_m$.

In particular, we are interested in the case of $p=\epsilon/n$.
Then the binomial distribution converges to the Poisson distribution as
${n \choose k} p^k (1-p)^{n-k} \to e^{-\epsilon}\frac{\epsilon^k}{k !}$ as $n \to \infty$.
Also, we have the condition $j\ge n/2-k$. 
Let the coefficients 
$a_m:= \frac{1}{\sqrt{n+1}}f(\frac{m}{n})$,
where
$f$ is a square-integrable smooth function on $[0,1]$ with $l^2$ norm $1$. 
We prove the following lemma.
\begin{lemma}\Label{Le1}
When $k,\ell$ is fixed and 
$a_m= \frac{1}{\sqrt{n+1}}f(\frac{m}{n})$,
$a_m d_{n, \frac{n}{2}-t, m,k,\ell}$
is approximately $\frac{1}{\sqrt{n+1}}(\sqrt{T_{t,k,\ell}} f) (\frac{m}{n})$, where
the operator $T_{t,k,\ell}$ is defined as
\begin{align}
&T_{t,k,\ell}\nonumber \\
:=& \sum_{u=\max(0,t-k+\ell)}^{\min(t,l)} 
{k-\ell \choose t-u} {\ell \choose u} Q^{2(t-u)+\ell}(I-Q)^{2u+k-\ell} .
\end{align}
\end{lemma}
\begin{proof}
We fix $k,l$ and take the limit $n\to \infty$.
Due to Eq.~\eqref{XI1}, we find that
\begin{align}
b_{n,m|k,\ell}\cong 
\sqrt{ (\frac{m}{n})^\ell (1-\frac{m}{n})^{k-\ell} }.
\end{align}
Hence, we discuss $d_{n,\frac{n}{2}-t,m,k,\ell}$.
For this aim, we employ Theorem 5.1.1 of \cite{Yanagida}.
Due to \cite[(2.1.4)]{Yanagida}, the probability $p(x|n,m,k,\ell)$ defined in \cite{Yanagida}
equals 
$d_{n,\frac{n}{2}-x,m,k,\ell}^2$.
Hence, Theorem 5.1.1 of \cite{Yanagida} guarantees that
\begin{align}
& d_{n,\frac{n}{2}-t,m,k,\ell}^2\nonumber  \\
\cong & 
\sum_{u=\max(0,t-k+\ell)}^{\min(t,l)} 
{k-\ell \choose t-u} {\ell \choose u} (\frac{m}{n})^{2(t-u)+\ell}(1-\frac{m}{n})^{2u+k-\ell} .
\end{align}
Since $d_{n,\frac{n}{2}-t,m,k,\ell}\ge 0 $ and $a_m= \frac{1}{\sqrt{n+1}}f(\frac{m}{n})$, we have
\begin{align}
a_m d_{n, \frac{n}{2}-t, m,k,\ell}
\cong \frac{1}{\sqrt{n+1}} (\sqrt{T_{t,k,\ell}} f) (\frac{m}{n}).
\end{align}
\end{proof}

Therefore, $v_n$ is approximated as the following:
\begin{align}
v_n \cong 
\sum_{k=0}^\infty
e^{-\epsilon}\frac{\epsilon^k}{k !}
\sum_{t=0}^{k}
\sum_{\ell=0}^k {k \choose \ell} 
\frac{1}{4n^2}\langle f| \sqrt{T_{t,k,\ell}}P^2\sqrt{T_{t,k,\ell}}|f\rangle.
\end{align}
When $f$ satisfies the Dirichlet boundary condition $f(0)=f(1)=0$,
$\sqrt{T_{t,k,\ell}} f$ also satisfies the Dirichlet boundary condition
so that the average error $v_n$ achieves the Heisenberg scaling.

Next, we consider the limiting distribution.
For $p=\epsilon/n$, using the same discussion with Eq.~\eqref{ZXO}, we have
\begin{align}
& 
P_{\theta,\frac{\epsilon}{n},n} \Big\{
\frac{a}{n} \le \hat\theta -\theta \le \frac{b}{n} \Big\}\nonumber \\
\cong &
\sum_{k=0}^\infty
e^{-\epsilon}\frac{\epsilon^k}{k !}
\sum_{t=0}^{k}
\sum_{\ell=0}^k {k \choose \ell} 
\int_{a}^{b}
|{\cal F}(\sqrt{T_{t,k,\ell}}f)(t)|^2 dt .\Label{ZXO2}
\end{align}
That is, we conclude that 
the Heisenberg scaling  in terms of limiting distribution can be achieved
even when $f$ does not satisfy the Dirichlet boundary condition.


Therefore, we arrive at the following main conclusion.
\thmglobal*

\section{Practical global estimation with one-qubit memory}\label{Sec:seq}
We have shown that the Heisenberg scaling 
can be achieved when the noise parameter $p$ behaves as $\epsilon/n$.
However, the method given in Appendix \ref{SEMA} requires a  complicated process so that
it may not be regarded practically implementable.
To address this problem, similar to Appendix \ref{NNP},
we propose a simple adaptive method by modifying the 
adaptive discrete phase estimation method by the paper \cite{Cleve}, which requires only an one-qubit memory.
It is known that the above discrete method perfectly estimates the unknown phase parameter when 
it is limited to the specific discrete subset, and the estimation error of this method 
when the unknown phase parameter does not belong to the discrete subset is  discussed in \cite{Cleve}.
However, it did not derive the limiting distribution nor the asymptotic behavior of the average error
even in the noiseless case
when the unknown phase parameter is subject to the uniform distribution on the continuous set.
Appendix \ref{S8-2} clarifies the above two issues in the noiseless case
after Appendix \ref{S8-1} introduces the above discrete method with one-qubit memory as our practical phase estimator.
Appendix \ref{S8-3} analyzes the noisy case by
modifying the analysis of the above noiseless case, showing that it achieves the Heisenberg scaling in terms of limiting distribution.

\subsection{Construction of the estimator}\Label{S8-1}
According to \cite{Cleve}, we construct the following adaptive estimator on an one-qubit memory that works globally
when $n=2^{N+1}-1$ applications are allowed. (This protocol is already presented in the main text;  We repeat it here for readers' convenience.)

\begin{protocolc}{1}\Label{P1}
In the first step, 
we prepare the input state $|+\rangle:=
\frac{1}{\sqrt{2}}(|0\rangle+|1\rangle)$,
and apply the unknown channel $\Lambda_{\theta,p}$ for $2^N$ times.
Then, we measure the final state in the basis
$\{|+\rangle, |-\rangle \}$ and set
$A_1=0, 1$ upon getting $|+\rangle, |-\rangle$ respectively.

In the second step, we prepare the input state $|+\rangle:=\frac{1}{\sqrt{2}}(|0\rangle+|1\rangle)$, and apply $\Lambda_{\theta,p}$ for $2^{N-1}$ times. Then, we apply $U_{-A_1 \pi/2 }$ depending on $A_1$. Then, we measure the final state in the basis $\{|+\rangle, |-\rangle \}$ and set $A_2=0, 1$ upon getting $|+\rangle, |-\rangle$ respectively.

Inductively, in the $k$-th step, 
we prepare the input state $|+\rangle:=
\frac{1}{\sqrt{2}}(|0\rangle+|1\rangle)$,
and apply  $\Lambda_{\theta,p}$ for $2^{N-k+1}$ times.
Then, we apply $U_{-A_1  2^{-k+1} \pi-A_2 2^{-k+2}\pi - \cdots -A_{k-1} 2^{-1}\pi  }$ depending on $A_1, \cdots,A_{k-1}$.
Then, we measure the final state in the basis
$\{|+\rangle, |-\rangle \}$ and set
$A_k=0, 1$ upon getting $|+\rangle, |-\rangle$ respectively.

We repeat the above up to the $(N+1)$-th step.
After the final step, depending on $A^{N+1}:=(A_1, \cdots, A_{N+1})$, we obtain the final estimate
$\hat{\theta}(A^{N+1}):= A_1  2^{-N} \pi-A_2 2^{-N+1}\pi + \cdots + A_k 2^{k-(N+1)}\pi
+ \cdots +A_{N}2^{-1} \pi +A_{N+1} \pi$. 
\end{protocolc}

\subsection{Noiseless case}\Label{S8-2}
First consider the case $p=0$.
When $\theta=\ell 2^{-N+1}\pi$ for an integer $\ell$,
the above method can identify 
$\theta$ with probability $1$ \cite{Cleve}.
However, when the true parameter $\theta$ does not take the above discrete values, the situation is more complicated.
For the analysis of this situation, we rewrite the above estimator. 
Let ${\cal H}$ be the Hilbert space spanned by 
$\{|x\rangle \}_{x=0}^{2^{N+1}-1}$.
We define the representation $F_\theta$ on ${\cal H}$
by
\begin{align}
F_\theta|x\rangle:= e^{i \theta x}|x\rangle.
\end{align}
Then, we consider the $n$-tensor product system ${\cal H}_1 \otimes \cdots\otimes {\cal H}_{N+1}$,
where each ${\cal H}_j$ is spanned by $|0\rangle, |1\rangle$.
Then, we define an isomorphism $V$ from 
${\cal H}_1 \otimes \cdots\otimes {\cal H}_{N+1}$
to ${\cal H}$ 
as follows:
\begin{align}
V: |x_1\rangle\cdots |x_{N+1}\rangle \mapsto \left| \sum_{j=1}^{N+1} 
x_j 2^{N+1-j} \right\rangle .
\end{align}

Therefore, the outcome of the above protocol has the same stochastic behavior
as the outcome of the following protocol.

\begin{protocol}\Label{P2}
Set the initial state $|\eta_{\rm uni}\rangle:=
2^{-(N+1)/2}
\sum_{x=0}^{2^{N+1}-1}|x\rangle$.
Then, apply the unitary $V^{-1} F_\theta$.
Finally, make measurements in the following way:
In the first step, measure the system ${\cal H}_1$ in the basis $\{|\pm \rangle\}$;
In the $k$-th step, measure the system ${\cal H}_{k}$ in the basis $\{|\pm \rangle\}$
after applying the unitary 
$U_{-A_1  2^{-k+1} \pi-A_2 2^{-k+2}\pi - \cdots - A_{k-1} 2^{-1}\pi  }$;
After the final step, the $(N+1)$-th step, we obtain the final estimate $\hat{\theta}(A^{N+1})$.
\end{protocol}

Applying $V$ to the measurement basis 
on ${\cal H}_1 \otimes \cdots\otimes {\cal H}_{N+1}$ in Protocol \ref{P2},
we obtain the measurement basis $\{ F_{y 2^{-(N+1)} \pi} | \eta_{\rm uni} \rangle\}_{y=0}^{2^{N+1}-1}$ on ${\cal H}$.
Hence, the outcome of the above protocol has the same stochastic behavior
as the outcome of the following protocol.

\begin{protocol}\Label{P3}
Set the initial state $| \eta_{\rm uni}\rangle$.
After applying the unitary $ F_\theta$,
make the measurement $\{ F_{y 2^{-(n+1)}\pi } | \eta_{\rm uni}\rangle\}_{y=0}^{2^{N+1}-1}$.
Then, our estimate $\hat{\theta}$ is set to be $y 2^{-(N+1)}$.
\end{protocol}

Now, under Protocol \ref{P3},
we assume that the unknown parameter $\theta$ is subject to the uniform distribution on $[0,2\pi)$.
The difference 
$Z:=\theta- \hat{\theta}$ is subject to the distribution with the following probability density function,
\begin{align}
&P_Z(z)\nonumber \\
:= &
\frac{1}{2\pi} \sum_{y=0}^{2^{N+1}-1} |\langle F_{y 2^{-(N+1)}\pi }^\dagger \eta_{\rm uni} | 
F_{z+y 2^{-(N+1)}\pi }\eta_{\rm uni} \rangle |^2
\nonumber \\
=&
\frac{1}{2\pi} \sum_{y=0}^{2^{N+1}-1} |\langle F_{-z}^\dagger \eta_{\rm uni} | F_0\eta_{\rm uni} \rangle |^2\nonumber \\
=&
\frac{2^{N+1}}{2\pi}  |\langle F_{-z}^\dagger \eta_{\rm uni} | F_0 \eta_{\rm uni}\rangle |^2.
\end{align}
The final term is the same as the probability density function 
of  difference between the estimate and the true parameter in the setting of Section \ref{S7-A}
with initial state $| \eta_{\rm uni}\rangle$.
That is, when we consider the uniform distribution for the unknown parameter $\theta$ and 
focus on the average,
our analysis is reduced to that in Section \ref{S7-A} with the group covariant estimator.
However, we stress that the group covariant estimator in Section \ref{S7-A} cannot be written in a form of an adaptive protocol with one-qubit memory.
Hence, to keep the above practical form of our estimator, 
we need to consider the averaged probability  
$\mathbb{E}_\theta [P_{\theta,0,2^{N+1}-1} \{ |\hat{\theta}-\theta| > c\}]$.
By using Eq.~\eqref{ZXO}, the tail probability is evaluated as
\begin{align}
&\lim_{N \to \infty}\mathbb{E}_\theta 
\Big[P_{\theta,0,2^{N+1}-1} \Big\{ |\hat{\theta}-\theta| > \frac{c}{n}\Big\}\Big]\nonumber \\
=&\lim_{N \to \infty}\mathbb{E}_\theta 
\Big[P_{\theta,0,2^{N+1}-1} \Big\{ |\hat{\theta}-\theta| > \frac{c}{2^{N+1}-1}\Big\}\Big]\nonumber \\
=&\int_{|t| > c} \frac{\sin^2 t }{t^2}\frac{dt}{2\pi},
\Label{EQ1}
\end{align}
where $\mathbb{E}_\theta$ denotes the average with respect to the uniform prior.
Notice that the variance of the above limiting distribution is not finite.
Using Eq.~\eqref{XO1}, we have 
\begin{align}
&\lim_{N \to \infty}\mathbb{E}_\theta [\sin^2(\hat{\theta}-\theta)] n
=2 \cdot \lim_{N\to \infty}\mathbb{E}_\theta 
[\sin^2(\hat{\theta}-\theta)] 2^n\nonumber  \\
=&2 \cdot \frac{1}{2}
=\frac{\Si(2\pi)}{ \pi }. \Label{XO2}
\end{align}
That is, 
the Heisenberg scaling  in terms of average error
cannot be attained even in the noiseless case.

\subsection{Noisy case}\Label{S8-3}
Next, we analyze the  case of non-zero $p$.
Let $\hat{A}_k$ 
be the outcome of the $k$-th step in Protocol \ref{P1} in the noiseless case.
Again let $X_k$ be the variable that describes the error in the outcome of the $k$-th step.
That is, when the outcome of the $k$-th step is flipped, $X_k=1$. Otherwise, it is zero.  (Note that we denote $X^{N+1}:=(X_1, \cdots, X_{N+1})$ like for $A$.)
Hence, we obtain the outcome $A_k=\hat{A}_k \oplus X_k$ in the $k$-th step.
The probability that the correct unitary $U_\theta^{\otimes 2^{N-k+1}}$
acts in the $k$-th step is $(1-p)^{2^{N-k+1}}$.
When the correct unitary $U_\theta^{\otimes 2^{N-k+1}}$ does not act in the $k$-th step,
the outcome $A_k$ of the $k$-th step is subject to the uniform distribution, which implies that 
the outcome $A_k$ of the $k$-th step equals $\hat{A}_k$ with probability $1/2$.
Therefore, 
the probability $P_{X_k}(1)$ is characterized as  
\begin{align}
P_{X_k}(1)=\frac{1}{2}\Big(1-(1-p)^{2^{N-k+1}}\Big). \Label{MO1}
\end{align}
Let $\hat{\theta}$ be the estimate.
We denote the probability distribution of $\hat{\theta}$ when the true parameter is $\theta$ 
by $P_{\hat{\theta}|\theta}$, which is given by
\begin{align}
&P_{\hat{\theta}|\theta}( \hat{\theta}(a^{N+1}) 
) \nonumber \\
=&
\sum_{x^{N+1} \in \mathbb{F}_2^{N+1}}
P_{X_1}(x_1)\cdots P_{X_{n+1}}(x_{N+1}) \nonumber \\
&\cdot |\langle F_{
 ( (a_1\oplus x_1)  2^{-N} + \cdots +(a_{N+1}\oplus x_{N+1}) )\pi
 }^\dagger \eta_{\rm uni} | 
F_{\theta}\eta_{\rm uni}\rangle |^2 \nonumber \\
=&
\sum_{x^{N+1} \in \mathbb{F}_2^{N+1}}
P_{X_1}(x_1)\cdots P_{X_{N+1}}(x_{N+1})\nonumber \\
&\cdot |\langle F_{\hat{\theta}(a^{N+1})+{\tau}(a^{N+1},x^{N+1})
 }^\dagger \eta_{\rm uni} | 
F_{\theta}\eta_{\rm uni}\rangle |^2 , \label{a}
\end{align} 
where 
\begin{align}
&\tau(a^{N+1},x^{N+1})\nonumber\\:=& 
((-1)^{a_1}x_1  2^{-N} + \cdots +(-1)^{a_{N+1}}x_{N+1})\pi  ,
\end{align} 
and
Eq.~\eqref{a} follows from the relation $a_k\oplus a_k=a_k + (-1)^{a_k}x_k$.

We also assume that $\theta$ is subject to the uniform distribution.
Hence, the joint distribution $P_{\hat{\theta},\theta}(\hat{\theta}_0,\theta_0)$ 
of $\hat{\theta}$ and $\theta$
is $ P_{\hat{\theta}|\theta_0 }(\hat{\theta}_0) \frac{1}{2\pi}$,
where $\hat{\theta}$ takes a discrete value and $\theta$ takes a continuous value.
Hence, the joint distribution of 
the difference $Z:=\theta- \hat{\theta}$ and $\hat{\theta}$ 
is given as
\begin{align}
P_{Z,\hat{\theta}}( z, \hat{\theta}_0)
= P_{\hat{\theta}|z+\hat{\theta}_0 }(\hat{\theta}_0) \frac{1}{2\pi}.
\end{align}
Thus, the difference $Z$ is subject to the distribution with the following probability density function;
\begin{align}
&P_Z(z)\nonumber \\
= &
\frac{1}{2\pi}
\sum_{a^{N+1} \in \mathbb{F}_2^{N+1}}
 P_{\hat{\theta}|z+\tau(a^{N+1}) }( \tau(a^{N+1}) ) 
\nonumber \\
= &
\frac{1}{2\pi}
\sum_{a^{N+1} \in \mathbb{F}_2^{N+1}}
\sum_{x^{N+1} \in \mathbb{F}_2^{N+1}}
P_{X_1}(x_1)\cdots P_{X_{N+1}}(x_{N+1})
\nonumber \\
&\cdot  |\langle F_{
\hat{\theta}(a^{N+1})+{\tau}(a^{N+1},x^{N+1})
 }^\dagger\eta_{\rm uni} | 
F_{z+\hat{\theta}(a^{N+1})
}\eta_{\rm uni}\rangle |^2
\nonumber \\
= &
\frac{1}{2\pi}
\sum_{a^{N+1} \in \mathbb{F}_2^{N+1}}
\sum_{x^{N+1} \in \mathbb{F}_2^{N+1}}
P_{X_1}(x_1)\cdots P_{X_{N+1}}(x_{N+1})
\nonumber \\
&\cdot  |\langle F_{
(   -z+ \tau(a^{N+1},x^{N+1})}^\dagger \eta_{\rm uni} | 
F_{0}\eta_{\rm uni}\rangle |^2\nonumber \\
= &
\frac{2^{N+1}}{2\pi}
\sum_{a^{N+1} \in \mathbb{F}_2^{N+1}}
\frac{1}{2^{N+1}}
\sum_{x^{N+1} \in \mathbb{F}_2^{N+1}}
P_{X_1}(x_1)\cdots P_{X_{N+1}}(x_{N+1})
\nonumber \\
&\cdot  |\langle F_{
-z+ \tau(a^{N+1},x^{N+1}) }^\dagger \eta_{\rm uni} | 
F_{0}\eta_{\rm uni}\rangle |^2.\Label{NOU}
\end{align}

Now, we define the random variable $Z_0$ subject to the probability density function 
$|\langle F_{z_0}^\dagger I |  \eta_{\rm uni}\rangle |^2 
\frac{d z_0}{2\pi} $.
Hence, Eq.~\eqref{NOU} guarantees that the difference $Z$ is characterized as 
\begin{align}
Z=Z_0 - \tau(A^{N+1},X^{N+1}),\Label{CA1}
\end{align}
where the binary variables $A_1, \cdots, A_{N+1} \in \mathbb{F}_2$ 
are independent binary variable subject to the uniform distribution
and the binary variable $X_k$ is subject to the distribution Eq.~\eqref{MO1}.

Now, we consider the case when $p=\frac{\epsilon}{n}=\frac{\epsilon}{2^{N+1}-1}$.
Hence, the probability $P_{X_k}(1)$ given in Eq.~\eqref{MO1} converges to 
$\frac{1}{2}(1- e^{-\epsilon 2^{-k}})$.
Hence, in the following, we consider that the binary variables
$X^{N+1}=(X_1, \cdots, X_k, \cdots, X_{N+1})$ are independently subject to the following distribution:
\begin{align}
P_{X_k}(1)=\frac{1}{2} \big(1-e^{-\epsilon 2^{-k}}\big),\quad
P_{X_k}(0)=\frac{1}{2} \big(1+e^{-\epsilon 2^{-k}}\big).
\Label{AJ1}
\end{align}
Hence,  $X_k$ with large $k$ can be ignored because  
$\frac{1}{2} \big(1-e^{-\epsilon 2^{-k}}\big)
\cong \epsilon 2^{-k-1}$ goes to zero as $k$ goes to infinity.
Also, the binary variables $\hat{A}^{N+1}=(\hat{A}_1, \cdots, \hat{A}_k, \cdots, \hat{A}_{N+1})$ are other independent binary variables subject to the uniform distribution.
So we have
\begin{align}
& \lim_{N \to \infty}\mathbb{E}_\theta \Big[P_{\theta,\frac{\epsilon}{n},2^{N+1}-1} 
\Big\{ \frac{a}{n}\le \hat{\theta}-\theta \le \frac{b}{n}\Big\}\Big] \nonumber \\
=&
\lim_{N \to \infty}
\mathbb{E}_{\hat{A}^{N+1},X^{N+1}}
\bigg[\int_{b- \zeta(\hat{A}^{N+1},X^{N+1})}^{a-  \zeta(\hat{A}^{N+1},X^{N+1})}
 \frac{\sin^2 y }{y^2}\frac{dy}{2\pi}  \bigg]\nonumber \\
=&
\lim_{n \to \infty}
\mathbb{E}_{\hat{A}^{n+1},X^{n+1}}
\bigg[\int_{b}^{a}
 \frac{\sin^2 (y+\zeta(\hat{A}^{N+1},X^{N+1})) }{(y+\zeta(\hat{A}^{N+1},X^{N+1}))^2}\frac{dy}{2\pi}\bigg] \nonumber \\
 =&\lim_{N \to \infty}
\mathbb{E}_{\hat{A}^{N+1},X^{N+1}}
\bigg[\int_{b}^{a}
 \frac{\sin^2 y }{(y+\zeta(\hat{A}^{N+1},X^{N+1}))^2}\frac{dy}{2\pi}\bigg],
 \Label{EY1}
\end{align}
where
$\zeta(\hat{A}^{N+1},X^{N+1}):= 
((-1)^{\hat{A}_1}2 X_1   + \cdots +(-1)^{\hat{A}_{N+1}}2^{N+1} X_{N+1})\pi $.
The above equation shows that 
the proposed estimator
achieves the Heisenberg scaling in terms of limiting distribution.

Also, the asymptotic behavior of the average error is calculated as
\begin{align}
&\lim_{N \to \infty}n \mathbb{E} [Z^2] 
-\frac{\Si(2\pi)}{ \pi } \nonumber \\
\stackrel{\text{(a)}}{=}&\lim_{N \to \infty}n \mathbb{E}
\Big[\Big(Z_0 - 
\tau(\hat{A}^{N+1},X^{N+1})
\Big)^2-Z_0^2\Big] \nonumber \\
\stackrel{\text{(b)}}{=}&\lim_{N \to \infty}
n \mathbb{E}
\Big[\Big( Z_0^2 +\sum_{k=1}^{N+1} X_k^2  2^{2k-2-2N} 
\pi^2  \Big)-Z_0^2\Big]\nonumber \\
=&\lim_{N \to \infty}2^{N+1} 
\Big( \sum_{k=1}^{N+1} (1-e^{-\epsilon 2^{-k}})  2^{2k-3-2N} \pi^2  \Big) \nonumber \\
=&\lim_{N \to \infty} 
 \sum_{k=1}^{N+1} (1-e^{-\epsilon 2^{-k}})  2^{2k-2-N} \pi^2 \nonumber \\
=&\lim_{N \to \infty}  \sum_{k=1}^{N+1} 
\sum_{l=1}^\infty\frac{(-1)^{l+1} \epsilon^l 2^{-kl} }{l !}
  2^{2k-2-N} \pi^2  \nonumber \\
=&\lim_{N \to \infty} \sum_{l=1}^\infty \sum_{k=1}^{N+1} 
\frac{(-1)^{l+1} \epsilon^l }{l !}
  2^{(2-l)k-2-N} \pi^2 \nonumber \\
=&\lim_{N \to \infty}
\sum_{l=1}^\infty\frac{(-1)^{l+1} \epsilon^l   2^{-2-N} \pi^2 }{l !}
  \frac{2^{(N+2)(2-l)}-2^{2-l}}{ 2^{2-l}-1 }
 \nonumber \\
\stackrel{\text{(c)}}{=}
&
\epsilon  2^{-2} \pi^2 
  \frac{2^{2}}{ 2-1 }
=\epsilon  \pi^2 ,
\end{align}
where $(a)$ follows from Eq.~\eqref{CA1} and Eq.~\eqref{XO2},
$(b)$ follows from the independence and the uniformity of $A_k$, and in addition,
$(c)$ holds because the terms in $\sum_{l=1}^\infty$ vanish when $l \ge 2$.

\end{document}